\documentclass[10pt]{IEEEtran}

\usepackage{amsmath,amssymb,amsfonts,amsthm}
\usepackage{graphicx}
\usepackage[table,pdftex,hyperref,svgnames]{xcolor}
\usepackage{bbold}
\usepackage{enumitem,xspace}
\usepackage{mathtools,bbm}

\usepackage{hyperref,doi}
\hypersetup{colorlinks=true, linkcolor=blue, breaklinks=true, urlcolor=blue, 
	citecolor=blue,}

\usepackage[all]{xy}
\usepackage{mathrsfs}
\usepackage{tikz}
\usepackage[utf8]{inputenc}
\usepackage[yyyymmdd]{datetime}
\usepackage{graphicx}
\usepackage{subcaption}
\usepackage{lipsum}
\usepackage{multicol} \usepackage{mathtools}
\usepackage{multirow}
\usepackage{threeparttable}

\usepackage[noend]{algorithmic}
\usepackage{algorithm}
\algsetup{indent=2em}

\newsavebox{\ieeealgbox}

\usepackage{multicol} \usepackage{mathtools}
\usepackage{version,xspace}
\usepackage{url,doi}

\newcommand\oprocendsymbol{\hbox{$\triangle$}}
\newcommand\oprocend{\relax\ifmmode\else\unskip\hfill\fi\oprocendsymbol}

\usepackage{enumitem}

\DeclareSymbolFont{bbold}{U}{bbold}{m}{n}
\DeclareSymbolFontAlphabet{\mathbbold}{bbold}

\newcommand{\real}{\mathbb{R}}

\DeclareMathOperator{\diag}{diag}

\renewcommand{\top}{\mathsf{T}} 

\newtheorem{theorem}{Theorem}

\newtheorem{lemma}[theorem]{Lemma}
\newtheorem*{lemma*}{Lemma}
\theoremstyle{definition}
\newtheorem{definition}[theorem]{Definition}
\newtheorem{assumption}{Assumption}

\newtheorem*{example*}{Example}

\newcommand{\suchthat}{\;\ifnum\currentgrouptype=16 \middle\fi|\;}
\newcommand{\scirc}{\raise1pt\hbox{$\,\scriptstyle\circ\,$}}

\usepackage{scalerel}
\usepackage{tikz}
\usetikzlibrary{svg.path}
\definecolor{orcidlogocol}{HTML}{A6CE39}
\tikzset{
	orcidlogo/.pic={
		\fill[orcidlogocol] 
		svg{M256,128c0,70.7-57.3,128-128,128C57.3,256,0,198.7,0,128C0,57.3,57.3,0,128,0C198.7,0,256,57.3,256,128z};
		\fill[white] svg{M86.3,186.2H70.9V79.1h15.4v48.4V186.2z}
		svg{M108.9,79.1h41.6c39.6,0,57,28.3,57,53.6c0,27.5-21.5,53.6-56.8,53.6h-41.8V79.1z
			
			M124.3,172.4h24.5c34.9,0,42.9-26.5,42.9-39.7c0-21.5-13.7-39.7-43.7-39.7h-23.7V172.4z}
		svg{M88.7,56.8c0,5.5-4.5,10.1-10.1,10.1c-5.6,0-10.1-4.6-10.1-10.1c0-5.6,4.5-10.1,10.1-10.1C84.2,46.7,88.7,51.3,88.7,56.8z};
	}
}
\newcommand\orcidicon[1]{\href{https://orcid.org/#1}{\mbox{\scalerel*{
				\begin{tikzpicture}[yscale=-1,transform shape]
				\pic{orcidlogo};
				\end{tikzpicture}
			}{|}}}}

\usepackage{blkarray}

\DeclareMathOperator{\E}{E}
\DeclareMathOperator{\cov}{cov}

\DeclareMathOperator{\supp}{supp}
\DeclareMathOperator{\argmin}{argmin}

\usepackage{tikz}
\usetikzlibrary{shapes,arrows}

\def\multiset#1#2{\ensuremath{\left(\kern-.3em\left(\genfrac{}{}{0pt}{}{#1}{#2}\right)\kern-.3em\right)}}

\newcommand{\rev}[1]{#1}
\newcommand{\monpaths}{P_m}

\begin{document}

	\title{Topology Inference with Multivariate Cumulants: \\ The M\"obius Inference 
	Algorithm}
	
	\author{Kevin D. Smith, 
		Saber Jafarpour, 
		Ananthram Swami, and
		Francesco Bullo 
		\IEEEcompsocitemizethanks{
				\IEEEcompsocthanksitem This work was supported in part by the U.S.\ 
	Defense 
	Threat Reduction Agency under grant
				HDTRA1-19-1-0017.
				\IEEEcompsocthanksitem Kevin D. Smith, Saber Jafarpour, and Francesco
				Bullo are with the Center
				of Control, Dynamical Systems and Computation, UC Santa Barbara, CA
				93106-5070, USA. {\tt \{kevinsmith, saber, bullo\}@ucsb.edu}
				\IEEEcompsocthanksitem Ananthram Swami is with the Army Research 
				Laboratory.}}
	

	%
	%

	\markboth{}%
	{Smith\MakeLowercase{\textit{et al.}}}
	%



	\maketitle
	
	\begin{abstract} 
		Many tasks regarding the monitoring, management, and design of communication 
		networks rely on knowledge of the routing topology. However, the standard 
		approach to topology mapping---namely, active probing with traceroutes---relies 
		on cooperation from increasingly non-cooperative routers, leading to missing 
		information. Network tomography, which uses end-to-end measurements of 
		additive link metrics (like delays or log packet loss rates) across monitor 
		paths, is a possible remedy. Network tomography does not require that routers 
		cooperate with traceroute probes, and it has already been used to infer the 
		structure of multicast trees. This paper goes a step further. We provide a 
		tomographic method to infer the underlying routing topology of an arbitrary set 
		of monitor paths \rev{using} the joint distribution of end-to-end measurements, 
		\rev{without making any assumptions on routing behavior.} Our approach, called 
		the \textit{M\"obius Inference Algorithm} (MIA), uses cumulants of this 
		distribution to quantify high-order interactions among monitor paths, and it 
		applies M\"obius inversion to ``disentangle'' these interactions. \rev{In 
		addition to MIA, we provide a more practical variant called \textit{Sparse 
		M\"obius Inference}, which uses various sparsity heuristics to reduce the number 
		and order of cumulants required to be estimated. We show the viability of our approach using 
		synthetic case studies based on real-world ISP topologies.}
	\end{abstract}
	

	%
	\IEEEpeerreviewmaketitle
	
	\section{Introduction}
         
 	Many tasks regarding the monitoring, management, and design of communication 
 	networks benefit from the network operator's ability to determine the routing 
 	topology, i.e., the incidence between paths and links in the network. During 
 	small-scale network failures, for example, routes may automatically switch, and it is 
 	important that the network operator \rev{has} knowledge of the new routing matrix. In 
 	the 
 	case of large-scale topology failures, inference of the routing topology is a crucial 
 	prelude to determining both the surviving network topology and the available services 
 	that remain. Peer-to-peer file-sharing networks are another example: nodes may want 
 	to know the routing topology so that they can select routes that have minimal overlap 
 	with existing routes, so as to avoid congestion and improve performance. Additional 
 	applications to the inference of dark networks and adversarial networks is obvious. 
 	Furthermore, the problem of optimal monitor placement relies on some knowledge of the 
 	network topology, and inference of the routing matrix provides topological 
 	information that could be used to bootstrap new end-to-end measurements.

	\paragraph*{Literature Review}
        
	Two main approaches are available for topology inference in communication 
	networks: using \textit{traceroutes}, and using \textit{network tomography} 	
	\cite{XZ-CP:11}. Traceroutes are the simplest and most direct approach, but they rely 
	on intermediate routers to cooperate by responding to traceroute packets. This 
	cooperation is becoming increasingly uncommon \cite{MHG-KS:09}, leading to 
	inaccuracies in traceroute-based topology mapping \cite{ML-YH-BH:08}. Some authors 
	\rev{have} modified traceroute approaches to account for uncooperative routers 
	\cite{BY-RV-FC-DW:03, XJ-WPKY-SHGC-YW:06, BH-ST-SS-TLP-AS:15}, using partial 
	traceroute results to over-estimate the topology, then applying heuristics and side 
	information to merge nodes. These approaches perform well on test 
	cases, but a rigorous method of selection among viable topologies would still be 
	desirable.
	
	Another approach to topology inference has started to emerge from the literature on 
	network tomography. Network tomography is the problem of inferring additive link 
	metrics (like delays or log packet loss rates) from end-to-end measurements; a nice 
	review is provided in \cite{MC-AOHIII-RN-BY:02}. Unlike traceroute approaches, 
	network tomography does not rely on intermediate routers to cooperate with traces. 
	Instead, it measures \rev{some \textit{metric} like} delay or log packet loss rate 
	between hosts, and it solves a linear inverse problem to infer 
	\rev{the values of these} metrics on each link. While most tomography literature 
	assumes that the \rev{routing matrix} is known, some authors have used tomographic 
	approaches to infer \rev{the} routing topology in special cases. \rev{In general, 
	these approaches to are based on a collection of statistics called \textit{path 
	sharing metrics} (PSMs), which are defined for each pair of host-to-host paths. The 
	PSM for a pair of paths is the sum of metrics across all links that are shared by 
	the two paths. A topology is then selected that explains all of the PSMs. 
 
	The tomographic approach was first applied to the single-source and 
	multiple-receiver setting to infer multicast trees. One of the first papers to adopt 
	this idea is \cite{SR-SM:99}, which uses joint statistics of packet loss between 
	pairs of receivers as a PSM. By repeatedly identifying the pair with greatest path 
	sharing, joining that pair into a ``macro-node,'' and re-computing the statistics, 
	the authors iteratively build the multicast tree from the bottom up. A few years 
	later, \cite{NGD-JH-FLP-DT:02} generalized this idea from 
	packet losses to other PSMs, including correlations between packet 
	delays between receiver pairs; and \cite{MC-RC-RN-MG-RK-YT:02} 
	accounted for measurement noise by moving the problem to a maximum likelihood 
	framework. Somewhat more recently, \cite{JN-HX-ST-YRY:09} re-considered the problem 
	of constructing a multicast tree from PSMs and provided new rigorous 
	and more-efficient algorithms. All of these papers use PSMs 
	for pairs of source-receiver paths to reconstruct the tree.

	Later work has extended tomographic topology inference from beyond multicast trees to 
	more general multiple-source, multiple-receiver problems. In \cite{MGR-MJC-RDN:06}, 
	the authors merge multicast trees to infer the topology with multiple 
	sources, under some ``shortest-path'' assumptions on the routing behavior---again 
	using PSMs. \cite{GB-ND-ME-KM:18} provides more general necessary and sufficient 
	conditions for when network inference is possible based on PSMs. Both of these papers 
	essentially assume shortest-path routing, an assumption which 
	is not always valid, for example, due to load balancing in the TCP layer 
	\cite{MGR-MJC-RDN:06}. This assumption also cannot accommodate more complex probing 
	paths, such as the two-way paths that emerge when a monitoring 
	endpoint pings another node.\footnote{We would like to thank an anonymous reviewer 
	for pointing out this possibility.} 
	
	Recent papers have also applied tomography to problems with uncertain (yet not 
	completely unknown) 
	topologies. In \cite{LM-TH-AS-DT-KKL:16}, the typical linear inverse problem from 
	tomography is replaced with a Boolean linear inverse problem, allowing the authors to 
	identify failed links from end-to-end data. Similarly, \cite{TH-AG-LM-KKL-AS-DT:17} 
	studies the problem of making network tomography robust to dynamics in the network 
	topology. The last two papers also deal with the problem of measurement design, i.e. 
	constructing the routing matrix to ensure identifiability. Neither of these two last 
	papers is concerned with inferring the routing matrix; however, they do represent 
	approaches outside of the PSM paradigm to gleaning topological information from 
	end-to-end data in a tomography setting.
	
	Another recent paper \cite{AAS-RKS-DT:19} introduced a new method for topology 
	inference, called ``OCCAM''. Like most of the other methods we have referenced, OCCAM 
	is based on PSMs; however, instead of algorithmically constructing the unique 
	topology that is consistent with the PSMs and routing assumptions, OCCAM solves an 
	optimization problem with an Occam's razor heuristic. The heuristic is not guaranteed 
	to find the correct network structure (unless the underlying network is a tree), but 
	the authors demonstrate good empirical performance. To our knowledge, OCCAM is the 
	only approach to truly general topology inference via network tomography, i.e., an 
	approach that does not require any assumptions on routing behavior (beyond 
	the fundamental assumption of stable paths between source-receiver pairs). }
	
	\paragraph*{Contributions}
	
	This paper provides \rev{another such} approach {to topology inference. We extend the 
	use of second-order \rev{PMSs} into higher-order statistics (i.e., statistics 
	involving \rev{more than two} paths), allowing us to relax any underlying assumptions 
	about the underlying topology.} Our method uses cumulants to quantify high-order 
	interactions between \rev{multiple} paths, then applies M\"obius inversion to 
	``disentangle'' these interactions, resulting in an encoding of the routing topology. 
	Our general approach, which we call the \textit{M\"obius Inference Algorithm} (MIA), 
	is a non-parametric method of reconstructing the routing matrix from multivariate 
	cumulants of end-to-end measurements, under mild assumptions. It does not require any 
	prior knowledge of the topology or distributions of link metrics, and works under 
	general routing topologies.
	
	The paper has three main contributions. First, we provide a novel application of 
	statistics and combinatorics to network tomography. We show that multivariate 
	cumulants of end-to-end measurements reveal interactions between the monitor paths 
	(in the form of overlapping links), and we demonstrate how M\"obius inversion can be 
	used to infer link-path incidence from these cumulants. Based 
	on these observations, we construct the \textit{M\"obius Inference 
	Algorithm} (MIA), which recovers a provably correct routing matrix from these 
	cumulants.
	
	Second, we \rev{adapt MIA to the more practical scenario in which} a dataset of 
	end-to-end measurements is available, instead of exact cumulants. \rev{This 
	``empirical'' variant of the routing inference algorithm applies} a hypothesis test 
	to every candidate column of the routing matrix, deciding based on the data whether 
	or not the column is present. This hypothesis testing is based on a	novel statistic, 
	and it works within any framework for location testing the mean of a 
	distribution. 
	
	\rev{
	Third, we create a more practical procedure, called \textit{Sparse M\"obius 
	Inference}, which modifies MIA using several sparsity heuristics. This procedure 
	minimizes the number of cumulants that need to be evaluated, restricts 
	cumulant orders to some user-specified limit, and reduces the time complexity of 
	the algorithm. The procedure also makes the inference more robust against measurement 
	noise, by replacing the exact M\"obius inversion formula with a lasso regression 
	problem.
	
	Finally, we use many numerical case studies, based on real-world Rocketfuel networks, 
	to evaluate the performance of Sparse M\"obius Inference. We study how the 
	performance depends on the underlying network, the number of monitor paths, the 
	sample size, and other parameters.}
	
%
	
	\paragraph*{Organization}
	
	\rev{This paper takes a didactic approach to introducing MIA and its sparse variant. 
	Section \ref{sect:prelims}} formally describes the communication network model and 
	key variables, \rev{provides a brief introduction to cumulants and $k$-statistics, 
	and} discusses 
	our three mild assumptions. \rev{Section \ref{sect:theory} considers the easiest 
	setting for 
	topology inference, wherein precise values for all of the necessary cumulants are 
	available without noise, so that we can focus on the core statistical and 
	combinatorial insights behind MIA. Section \ref{sect:data} then replaces the 
	precise cumulant values with noisy measurements. Then Section \ref{sect:sparse} 
	replaces MIA 
	altogether with the more practical Sparse M\"obius Inference procedure, which allows 
	the user to cap the order of cumulants they are willing to estimate. Finally, Section 
	\ref{sect:results} provides an overview of our numerical results and evaluation. 
	\textit{The	full set of numerical results, as well as all proofs of theoretical 
	results, are contained in appendices in the supplementary file.}}

	\section{Modeling and Preliminaries}
	\label{sect:prelims}

	\subsection{Model}
	\label{sect:model}
	
	We consider a network on a (possibly directed) graph $G$ \rev{with a set of links} 
	$\rev{L} = \{\ell_1, \ell_2, \dots, 
	\ell_m\}$. Every link is associated with an additive 
	link metric, like a time delay or log packet loss rate. We will 
	refer to these metrics simply as ``delays,'' although other metrics are possible. 
	
	For each link, there is a \textit{link delay variable} $U_\ell$, which is a random 
	variable representing the amount of time that a unit of traffic requires to traverse 
	the link. Link delays are not measured directly. Instead, we will infer properties of 
	these variables from cumulative delays across certain \rev{simple paths in $G$, 
	called \textit{monitor paths}. Let $\rev{\monpaths}$ be a set of $n$ monitor paths.} 
	Each $p \in \rev{\monpaths}$ is associated with a \textit{path delay variable}
	\begin{equation}
		V_p = \sum_{\substack{\ell \in \rev{L} ~\text{s.t.} \\p ~\text{traverses}~ 
		\ell}} U_\ell, 
		\qquad \forall p \in \rev{\monpaths}
		\label{eq:sum}
	\end{equation}
	which is the total delay experienced by a unit of traffic along the path $p$. If we 
	define a \rev{random} vector of link variables $\mathbf U = \begin{pmatrix} 
	U_{\ell_1} & U_{\ell_2} & \cdots & U_{\ell_m} \end{pmatrix}^\top$ and a \rev{random} 
	vector of path variables $\mathbf V = \begin{pmatrix} V_{p_1} & V_{p_2} & \cdots & 
	V_{p_n} \end{pmatrix}$, then we can write \eqref{eq:sum} in the form
	\begin{equation}
		\mathbf V = \mathbf R \mathbf U
		\label{eq:v}
	\end{equation}
	using a \textit{routing matrix} $\mathbf R \in \{0, 1\}^{n \times m}$, where $r_{p 
	\ell} = 1$ if and only if $p$ traverses the link $\ell$. \rev{We stress that we do 
	not make any assumptions about the nature of these monitor paths or the underlying 
	routing behavior. They may be one-way paths between monitoring endpoints, two-way 
	paths from a ping to a node and back, or both. The paths do not have to reflect 
	shortest-path routing.}
	
	We suppose that an experimenter is capable of measuring path delays $V_p(t)$ for each 
	monitor path $p$, at many sample times $t$. The experimenter has no 
	prior knowledge about the link variables $U_\ell$ and does not know the routing 
	matrix $\mathbf R$. Importantly, we make the simplifying assumption in this paper 
	that link delays are spatially and temporally independent, i.e., $U_\ell(t)$ 
	and $U_{\ell'}(t')$ are statistically independent unless $\ell = \ell'$ and $t = t'$. 
	This assumption is \rev{fundamental} in the network tomography literature 
	\cite{XZ-CP:11, MC-AOHIII-RN-BY:02, NGD-JH-FLP-DT:02, MC-RC-RN-MG-RK-YT:02, 
	JN-HX-ST-YRY:09, MGR-MJC-RDN:06}.

	\subsection{Preliminaries and Notation}
	\label{sect:notation}
	
	\paragraph*{General Notation} \; { Let $\mathbb Z_{\ge 0}$ and 
	$\mathbb Z_{> 0}$ denote the sets of non-negative and positive
          integers, respectively.}  Given a set $S$ and an integer $\rev{i} \le
        |S|$, let the binomial $ \binom{S}{\rev{i}} = \left\{S' \subseteq S :
        |S'| = \rev{i} \right\}$ denote the collection of all $\rev{i}$-element subsets
        of $S$. { Given $\rev{i}, n \in \mathbb Z_{\ge 0}$, let
          $\multiset{n}{\rev{i}}$ denote the number of $\rev{i}$-element multisets chosen
          from $n$ distinct elements.} Given two { ordered and}
        countable sets $X \subseteq Y$, define the \textit{characteristic
          vector} $\rev{\chi(X, Y)} \in \{0, 1\}^{|Y|}$ \emph{of $X$ in $Y$} by
        $\rev{\chi_i(X, Y)} = 1$ if and only if $y_i \in X$. \rev{Given any function $f: 
        X \to \real$, the \textit{support} of the function $\supp(f)$ is the subset of 
        elements $x \in X$ such that $f(x) \ne 0$.}

	\paragraph*{Multi-Indices} \rev{A \textit{multiset} is a set that allows for repeated 
	elements. A multiset can be represented by a \textit{multi-index}, which is a 
	function $\alpha: S \to \mathbb Z_{\ge 0}$ that maps each element of
	$S$ to its multiplicity in the multiset}. The \textit{support} of a 
	multi-index is the set of elements 
	with positive multiplicity, i.e., $\supp(\alpha) = \left\{s \in S : \alpha(s) \ge 1
	\right\}$. The \textit{size} of a multi-index is its total multiplicity: $|\alpha| = 
	\sum_{s \in S} \alpha(s)$. If $S$ is an ordered set with $n$ elements (e.g., if $S$ 
	consists of elements of a vector), then multi-indices on $S$ are naturally 
	represented as vectors $\alpha \in \mathbb Z_{\ge 0}^n$; in this case, 
	we will use multi-indices on $S$ and vectors in $\mathbb Z_{\ge 0}^n$ 
	interchangeably. \rev{For example, for  $S = \{a, b, c, d\}$, the multi-index 
	corresponding to the multiset $\{a, b, b, d, d, d\}$ can be represented by the vector 
	$\begin{pmatrix}1 & 2 & 0 & 3 \end{pmatrix}^\top$, using an alphabetic ordering of 
	$S$.}

	
	\paragraph*{Link Sets}
	Throughout this paper, we  make use of two maps from sets of monitor paths to 
	sets of links. \rev{Recall that} $\mathbf R \in \{0, 1\}^{n \times m}$ is the routing 
	matrix. For each $\rev{P \subseteq \monpaths}$, we define the \textit{common link 
	set} \rev{$C: 2^{\monpaths} \to 2^{L}$} by
	\begin{equation}
	\rev{C(P) = \{ \ell \in L : r_{p \ell} = 1, \; \forall p \in P\}}
	\label{eq:L}
	\end{equation}
	and the \textit{exact link set} \rev{$E: 2^{\monpaths} \to 2^{L}$} by
	\begin{equation} 
	\rev{E(P) = \{ \ell \in L: r_{p\ell} = 1, \; \forall p \in P ~\text{and}~ r_{p\ell} = 
	0, \; \forall p \notin P\}}
	\label{eq:M}
	\end{equation}
	The common link set \rev{$C(P)$} contains all links that are utilized
	by every path in $P$. The exact link set is more strict: \rev{$E(P)$}
	consists of links that are utilized by every path in $P$
	\textit{and} that are not utilized by any path outside of $P$.
	Neither of these maps are known \textit{a priori}. It is worth noting that the 
	exact link set contains all of the information of the routing matrix, since 
	\rev{$E(P)$} is nonempty if and only if the characteristic vector \rev{$\chi(P, 
	\monpaths)$} is a column of $\mathbf R$.

	\rev{As an example, consider the following routing matrix encoding 8 monitor paths 
	that utilize 8 links:
	
	\begingroup
	\setlength\arraycolsep{3pt}
	\small
	\[
		\mathbf R = \begin{pmatrix}
			1     & 0     & 0     & 0     & 0     & 0     &0     &0 \\
			0     & 0     & 0     & 0     & 1     & 1     &0     &0 \\
			0     & 0     & 0     & 1     & 0     & 1     &0     &1 \\
			0     & 0     & 1     & 1     & 0     & 1     &1     &0 \\
			0     & 1     & 1     & 1     & 0     & 0     &0     &0 \\ 
			0     & 1     & 0     & 0     & 1     & 0     &1     &0 \\
			0     & 1     & 1     & 0     & 1     & 0     &0     &1 \\
			0     & 1     & 0     & 0     & 0     & 1     &1     &0
		\end{pmatrix}
	\]
	\normalsize
	\endgroup
	In this example, $C(\{p_1\}) = E(\{p_1\}) = \{\ell_1\}$, since column~1 is the only 
	column with a nonzero first entry, and all other entries in the column are zero. 
	Furthermore, $C(\{p_3, p_7\}) = E(\{p_3, p_7\}) = \{\ell_8\}$, since column 8 is the 
	only column with a nonzero third and seventh entry, and all other entries are zero. 
	But $C$ and $E$ are not always equal: $C(\{p_5, p_6\}) = \{p_2\}$, but column 2 
	contains other nonzero entries as well, so $E(\{p_5, p_6\}) = \emptyset$. Multiple 
	common links are also possible, e.g., $C(\{p_6, p_7\}) = \{\ell_2, \ell_5\}$.
}

	\subsection{Cumulants and $k$-Statistics}
	
	Cumulants are a class of statistical moments, which extend the
        familiar notions of mean and covariance to higher orders. A good
        introduction is provided in \cite{PM-JK:09}; \rev{we provide a
          quick background here. Given a random variable $X$, define the
          \textit{cumulant generating function}
	\[
		K(t) = \log \E[e^{t X}] = \kappa_1 t + \frac{\kappa_2}{2!} t^2 + 
		\frac{\kappa_3}{3!} t^3 + \cdots
	\] 
	which admits a Taylor expansion for some sequence of coefficients $\kappa_1, 
	\kappa_2, \kappa_3, \dots$. These 
	coefficients are defined as the \textit{cumulants} of the random variable $X$. The 
	first three cumulants are identical to central moments: $\kappa_1$ is the mean of 
	$X$, $\kappa_2$ is the variance, and $\kappa_3 = \E[(X - \E[X])^3]$. For orders 
	four and higher, the relationship between cumulants and central moments is
	increasingly complicated. Table~\ref{table:cumulants} provides some examples of 
	common distributions whose cumulants have closed-form expressions. Given a random 
	variable $X$ and an integer $i \in \mathbb Z_{> 0}$, we let 
	$\kappa_i(X)$ denote the $i$th cumulant of $X$.
	
	\begin{table}
	\rev{
		\centering
		\begin{tabular}{lll}
			Distribution & Parameters & Cumulants \\
			\hline
			Normal & $\mu, \; \sigma^2$ & $\kappa_1 = \mu, \; \kappa_2 = \sigma^2, 
			\; \kappa_i = 0 ~\text{for}~ i \ge 3$ \\
			Exponential & $\lambda$ & $\kappa_i = \lambda^i (i - 1)! ~\text{for}~ i \ge 
			1$ \\
			Gamma & $\alpha, \; \beta$ & $\alpha \beta^{-i} (i - 1)! ~\text{for}~ i \ge 1$
		\end{tabular}
		\caption{Cumulants of some common univariate distributions.}
		\label{table:cumulants}
	}
	\end{table}

	Multivariate cumulants are an extension of cumulants to joint distributions. 
	Given some jointly-distributed random variables $X_1, X_2, \dots, X_n$, the cumulant 
	generating function is
	\[
		K(\mathbf t) = \log \E[e^{t_1 X_1 + \cdots + t_n X_n}] = \sum_{\alpha} 
		\frac{\kappa_\alpha}{|\alpha|!} \mathbf t^\alpha
	\]
	where the sum in the Taylor expansion occurs over all multi-indices $\alpha$ on the 
	set of integers $\{1, 2, \dots, n\}$, and $\mathbf t^\alpha$ denotes the product 
	$t_1^{\alpha(1)} t_2^{\alpha(2)} \cdots t_n^{\alpha(n)}$. Collecting $X_1, X_2, 
	\dots, X_n$ into the random vector $\mathbf X = \begin{pmatrix} X_1 & X_2 & \cdots & 
	X_n \end{pmatrix}^\top$, we use either the compact notation $\kappa_\alpha(\mathbf 
	X)$ or expanded notation $\kappa_\alpha(X_1, X_2, \dots, X_n)$ to represent the 
	multivariate cumulant of the joint distribution that corresponds to the 
	multi-index $\alpha$. If $\alpha$ is the 
	multi-index of all ones, we  drop the subscript and use the shorthand notation 
	$\kappa(X_1, X_2, \dots, X_n)$. We  also refer to the \textit{order} of a 
	cumulant as the size $|\alpha|$ of its multi-index.
	
	First-order multivariate 
	cumulants are means: if $\alpha$ has all zero multiplicites 
	except $\alpha(i) = 1$, then $\kappa_\alpha(\mathbf X) = \E[X_i]$. Second-order 
	multivariate cumulants are covariances: if $\alpha$ has all zero multiplicities 
	except $\alpha(i) = \alpha(j) = 1$, then $\kappa_\alpha(\mathbf X) = \cov(X_i, X_j)$. 
	If instead $\alpha(i) = 2$ with all other multiplicities zero, then 
	$\kappa_\alpha(\mathbf X) = \text{Var}(X_i)$. We will also make use of two general 
	properties of multivariate cumulants:
	\begin{enumerate}
		\item \textit{Multilinearity.} If $Y$ is a random variable independent from $X_1, 
		X_2, \dots, X_n$, then 
		\begin{align*}
			& \kappa_\alpha(X_1, \dots, X_i + Y, \dots, X_n) =\\
			& \qquad \kappa_\alpha(X_1, \dots, X_i, \dots, X_n) + \kappa_\alpha(X_1, 
			\dots, 
			Y, \dots, X_n)
		\end{align*}
		for any index $i$ and multi-index $\alpha$.
		\item \textit{Independence.} If any pair $X_i, X_j$ of the random variables $X_1, 
		X_2, \dots, X_n$ are independent, and $\alpha(i)$ and $\alpha(j)$ are both 
		non-zero, then $\kappa_\alpha(\mathbf X) = 0$.
	\end{enumerate}
}
	\rev{Cumulants can be computed analytically from joint distributions using the 
	generating function, but for unknown distributions, they must be estimated from 
	samples. Given an i.i.d. sample $\mathbf x_1, \mathbf x_2, \dots, \mathbf x_N \in 
	\real^n$ from $\mathbf X$, the \textit{$k$-statistic} $k_\alpha(\mathbf x_1, \mathbf 
	x_2, \dots, \mathbf x_n)$ is defined as the minimum-variance unbiased estimator of 
	$\kappa_\alpha(\mathbf X)$. The first and second-order $k$-statistics are 
	sample means and sample covariances, but higher-order $k$-statistics quickly 
	become more complex. We refer the reader to \cite{EDN-GG-DS:09} and \cite{KDS:2020b} 
	for a discussion of how general $k$-statistics are derived. For the purpose of this 
	paper, it suffices to note that} software packages are available to compute 
	\rev{$k$-statistics from samples}, both in R \cite{EDN-GG:19} and \rev{our own Python 
	library} \cite{KDS:20}.    
        
	
	\subsection{Assumptions}
	\label{sect:ass}
	
	At various points throughout the paper, we will invoke three closely-related 
	assumptions regarding the routing matrix and link delay cumulants. The first 
	assumption requires that $\mathbf R$ has no repeated columns:
	\begin{assumption}[Distinct Links] \label{ass:1}
		No two links are traversed by precisely the same set of paths in 
		\rev{$\monpaths$}; i.e., no two columns of $\mathbf R$ are identical; i.e., 
		$|\rev{E(P)}| \in \{0, 1\}$ for all \rev{$P \subseteq \monpaths$}. 
	\end{assumption}
	\noindent
	This assumption is common in the network tomography literature. If $\ell, 
	\ell' \in \rev{L}$ are used by precisely the same set of monitor paths, then the 
	link delays $U_\ell, U_{\ell'}$ will only show up in path delays through their sum 
	$U_{\ell} + U_{\ell'}$. Due to this linear dependence, complete network tomography is 
	impossible when Assumption \ref{ass:1} is violated, since $\mathbf R$ will be 
	rank deficient.
	
	The second assumption requires that link delays have nonzero cumulants:  
	\begin{assumption}[Nonzero Cumulants] \label{ass:2}
		For all $\ell \in \rev{L}$, and for all $\rev{i} = 2, 
		3, \dots, n$, \rev{the delay cumulant is nonzero:} $\kappa_{\rev i}(U_\ell) 
		\ne 0$. 
	\end{assumption}
	\noindent 
	\rev{For most practical purposes, one can think of Assumption \ref{ass:2} as meaning 
	that no link delay distribution is normally distributed. Non-normality is
	a necessary condition for the assumption to hold, since the normal distribution has 
	zero-valued cumulants for orders 3 and higher. Non-normality is not technically a 
	sufficient condition, since it is theoretically possible for a distribution to have 
	zero cumulants at some orders, but these cases are not common. In fact, the normal 
	distribution is the only distribution with a finite number of nonzero cumulants 
	\cite{PM-JK:09}. If link delays are known to be non-normally distributed, we consider 
	this to be a weak assumption.}  
        
	Finally, the third assumption requires that certain \textit{sums} of link delays have 
	nonzero cumulants:
	\begin{assumption}[Nonzero Common Cumulants] \label{ass:3}
		For all \rev{$P \subseteq \monpaths$}, and for all $\rev{i} = 2, 3, \dots, n$, if 
		$\rev{C(P)}$ is nonempty, then $\sum_{\ell \in \rev{C}(P)} \kappa_{\rev 
		i}(U_\ell) \ne 0$.  
	\end{assumption} 
	\noindent
	In other words, if all paths in $\rev{P \subseteq \monpaths}$ share a collection of 
	common links $\rev{C(P)}$, the delay cumulants on these common links should not 
	cancel out \rev{by summing to zero}. This is also a weak assumption, since 
	such a cancellation is very unlikely. \rev{In fact, many families of distributions 
	supported on $\real_{> 0}$ (including exponential and gamma distributions) have 
	strictly positive cumulants at all orders, in which case Assumption \ref{ass:3} is 
	satisfied automatically.} 
		
	\section{\rev{Theoretical Foundations}}
	\label{sect:theory} 
	
	We now proceed with our main theoretical contribution: a simple algorithm to 
	infer the routing matrix from multivariate cumulants of path latencies. \rev{The 
	purpose of this section is to state the underlying theoretical principles of MIA, so 
	we will temporarily assume that exact values for multivariate cumulants of the path 
	delay vector $\mathbf V$ are available. In reality, the experimenter seldom knows 
	these exact values and must estimate them via $k$-statistics instead, but this 
	requires some extra statistical treatment that we defer to Sections \ref{sect:data} 
	and \ref{sect:sparse}. For now, we will assume exact cumulant values to focus on the 
	discrete mathematics that underpin MIA. }
		
	
\rev{
	MIA works by identifying which exact link sets $E(P)$ are nonempty, since these 
	correspond precisely to columns of $\mathbf R$ (via the characteristic vector of 
	$P$). The sizes of the exact link sets are not directly observable, but they can be 
	inferred from the sizes of the common link sets. From \eqref{eq:L} and \eqref{eq:M}, 
	we can see that exact and common link sets are related by
	\[
		E(P) = C(P) \setminus \bigcup_{p' \notin P} C(P \cup \{p'\}).
	\]
	We can count the size of the union using the inclusion-exclusion principle:
	\begin{equation}\label{this:inclusion0-exclusion}
		\left| \bigcup_{p' \notin P} C(P \cup \{p'\}) \right| = \sum_{Q \supset P} 
		(-1)^{|Q| - |P| + 1} |C(Q)|.
	\end{equation}        
	Since $C(P \cup \{p'\}) \subseteq C(P)$ for all $p'$, we can use
        the inclusion-exclusion formula~\eqref{this:inclusion0-exclusion}
        to find the size of the exact link set as a function of the sizes
        of the common link sets:
	\begin{equation}
		|E(P)| = \sum_{Q \supseteq P} (-1)^{|Q| - |P|} |C(Q)|
		\label{eq:main}
	\end{equation}
	If we could somehow evaluate the number of common links shared by any set 
	of monitor paths, we could use the inclusion-exclusion principle to compute any 
	$|E(P)|$, from which we could reconstruct the routing matrix. 
	
	Unfortunately, counting the number of common links is typically infeasible in a 
	tomography setting. But the relationship in \eqref{eq:main} actually holds 
	for \textit{any} additive measure of link sets, not just cardinality, and some 
	additive measures can be inferred directly from end-to-end path data. For example, if 
	``$|C(Q)|$'' represents the sum of delay variances $\text{Var}(U_\ell)$ for each link 
	in $C(Q)$, then \eqref{eq:main} yields the sum of delay variances across links in 
	$E(P)$, which is nonzero if and only if $E(P)$ is nonempty. This sum of 
	delay variances across common links can be inferred from path delay data---at least 
	for pairs of monitor paths $p, p'$, the covariance $\text{cov}(V_p, V_{p'})$ is equal 
	to the sum of delay variances for each shared link in $C(\{p, p'\})$. For larger path 
	sets, we require higher-order statistics---like multivariate cumulants---to measure 
	``$|C(P)|$''.
	
	Having conveyed some of the core ideas behind MIA, we are ready to present the 
	algorithm itself and examine it with more theoretical rigor.} The algorithm occurs in 
	three stages:
	\begin{enumerate}
		\item
		\textit{Estimation}. Estimate a vector of multivariate cumulants of path 
		latencies. This vector contains information about the links that are common to 
		any given collection of paths. (The label ``estimation'' is a misnomer in the 
		context of \rev{this section}, wherein cumulants are known precisely, but it will 
		make more sense when we consider the ``data-driven'' version of the algorithm.)
		\item
		\textit{Inversion}. Apply a M\"obius inversion transformation to this vector of 
		estimates. The vector resulting from this transformation contains the routing 
		matrix, under a simple encoding. The transformation is linear, so this step can 
		be viewed as a matrix-vector multiplication.
		\item
		\textit{Reconstruction}. Decode the transformed vector, thereby reconstructing 
		the routing matrix.
	\end{enumerate} 
	
	\begin{algorithm}
		\caption{M\"obius Inference Algorithm \rev{(MIA)}}
		\label{alg:mia-dist}
		\begin{algorithmic}[1]
			\REQUIRE Joint distribution of path delays $\mathbf V$ 
			\ENSURE Routing matrix $\hat{\mathbf R}$
			\STATE \COMMENT{Estimation stage:}
			\label{line:estimation}
			\STATE Initialize undefined function $f_n: 2^{\rev{\monpaths}} \to \mathbb R$
			\FOR{$\rev{P \subseteq \monpaths}$} 
			\STATE Define \rev{$\alpha$} as any multi-index on $\rev{\monpaths}$ such 
			that $\supp(\alpha) = P$ and $|\alpha| = n$ \label{ln:midx}
			\STATE $f_n(P) \leftarrow \kappa_{\rev{\alpha}}(\mathbf V)$ \label{ln:deff}
			\ENDFOR
			\STATE \COMMENT{Inversion stage:}
			\label{line:inversion}
			\STATE Initialize undefined function $g_n: 2^{\rev{\monpaths}} \to \mathbb R$
			\FOR{$\rev{P \subseteq \monpaths}$}
			\STATE $g_n(P) \leftarrow \sum_{Q \supseteq P} (-1)^{|Q| - |P|} f_n(Q)$
			\ENDFOR
			\STATE \COMMENT{Reconstruction stage:}
			\label{line:recon}
			\STATE Initialize empty matrix $\hat {\mathbf R} \in \mathbb R^{n \times 0}$
			\FOR{\rev{$P \subseteq \monpaths$}}
			\rev{\IF{$g_n(P) \ne 0$} \label{ln:gnnz}
			\STATE $\hat {\mathbf R} \leftarrow \begin{pmatrix} \hat R & \rev{\chi(P, 
			\monpaths)} \end{pmatrix}$
			\ENDIF}
			\ENDFOR
			\RETURN $\hat {\mathbf R}$
		\end{algorithmic}
	\end{algorithm}

	\begin{theorem}[Analysis of \rev{MIA}] \label{thm:mia-dist}
		Consider the application of Algorithm \ref{alg:mia-dist} to a joint distribution 
		of path delays $\mathbf V = \begin{pmatrix} V_{p_1} & V_{p_2} & \cdots & V_{p_n} 
		\end{pmatrix}^\top$. Let $\mathbf R \in \{0, 1\}^{n \times m}$ be the true 
		underlying routing matrix, and let $\mathbf U = \begin{pmatrix} U_{\ell_1} & 
		U_{\ell_2} & \cdots & U_{\ell_m} \end{pmatrix}^\top$ be the underlying link 
		delays, so that $\mathbf V = \mathbf R \mathbf U$. The following are true:
		\begin{enumerate} 
			\item \label{item:mia-dist-terminates}
			The algorithm terminates and returns a matrix $\hat{\mathbf R} \in \{0, 
			1\}^{n \times \hat m}$ for some $\hat m \in \mathbb Z_{\ge 0}$, in $O(2^n)$ 
			time.
			\item \label{item:mia-dist-estimation}
			By line \ref{line:inversion}, the map $f_n: 2^{\rev{\monpaths}} \to \mathbb 
			R$ satisfies the following property:
			\begin{equation} \label{eq:f}
				f_n(P) = \sum_{\ell \in \rev{C}(P)} \kappa_n(U_\ell), \qquad \forall 
				\rev{P \subseteq \monpaths}
			\end{equation} 
			\item \label{item:mia-dist-inversion}
			By line \ref{line:recon}, the map $g_n: 2^{\rev{\monpaths}} \to \mathbb R$ 
			satisfies the following property:
			\begin{equation} \label{eq:g}
				g_n(P) = \sum_{\ell \in \rev{E}(P)} \kappa_n(U_\ell), \qquad \forall 
				\rev{P \subseteq \monpaths}
			\end{equation}
			\item \label{item:mia-dist-recon}
			Every column of $\hat{\mathbf R}$ is also a column of $\mathbf R$. 
			Furthermore, \rev{under} Assumptions \ref{ass:1} and \ref{ass:2}, $\mathbf R$ 
			and $\hat{\mathbf R}$ are equivalent (up to a permutation of columns).  
		\end{enumerate} 
	\end{theorem}

	Statement \ref{item:mia-dist-terminates} is obvious from inspection of the algorithm, 
	so we will focus on proving the remaining three statements, which fall neatly into 
	the three stages (estimation, inversion, and reconstruction) of the algorithm. In the 
	following subsections, we will analyze each of these three stages. 

	\subsection{Estimation Stage}
	\label{sect:estim}
	
	The purpose of the estimation stage is to collect a vector of high-order statistics 
	of path delays. These statistics are carefully chosen so that they contain 
	information about the routing topology. The title of ``estimation'' for this stage 
	will be more appropriate in the next subsection, when we must estimate these 
	statistics from data (rather than compute them analytically from a known 
	distribution). 
	
	In the estimation stage, we gather a vector of multivariate path delay 
	cumulants for every path set $\rev{P \subseteq \monpaths}$. The multivariate 
	cumulants that we select for each path set are based on representative 
	multi-indices:
	
	\begin{definition}[Representative Multi-Indices]
		Let $\rev{P \subseteq \monpaths}$, and let $\rev{i} \ge |P|$ be an integer. An 
		$\rev{i}$th-order \textit{representative multi-index} of $P$ is any
		multi-index $\alpha$ on $\rev{\monpaths}$ such that $\supp(\alpha) 
		= P$ and $|\alpha| = i$. We use the notation $\rev{A}_{i, P}$ to denote the 
		set of all $i$th-order representative multi-indices of $P$.
	\end{definition}
	
	\noindent
	We will now collect a vector of path delay cumulants, with one entry corresponding to 
	each set of monitor paths in $2^{\rev{\monpaths}}$:
	
	\begin{definition}[Common Cumulant] \label{def:f}
		Let \rev{$i$} be a positive integer. For each $\rev{P \subseteq \monpaths}$, let 
		$\alpha$ be any \rev{$i$}th-order representative multi-index of $P$. The 
		\rev{$i$}th-order \textit{common cumulant} is the \rev{map} $f_{\rev{i}}: 
		2^{\rev{\monpaths}} \to \mathbb R$ with entries
		\begin{equation}
		f_{\rev i}(P) = \kappa_{\alpha}(\mathbf V), \qquad \forall \rev{P \subseteq 
		\monpaths}
		\label{eq:f-def}
		\end{equation}
		
	\end{definition}

	\noindent
	Careful readers will also note that we refer to ``the'' common cumulant, rather than 
	``a'' common cumulant, which would seem more appropriate, given the many choices of 
	representative multi-indices. But the value of the common cumulant 
	is independent of the particular choice of representative multi-index---regardless of 
	which representative 
	multi-index we choose, it is always the sum of  univariate cumulants across links 
	that are traversed by every path in $P$. Broadly speaking, the \rev{value} of 
	$\rev{f_i(P)}$ contains information about which links are common to every path in 
	$P$.  
	
	\begin{lemma}[Properties of the Estimation Stage] \label{lem:mia-dist-estimation}
		The following are true:
		\begin{enumerate}
			\item \label{item:rep-midx-count}
			Let $\rev{P \subseteq \monpaths}$. If $\rev{i} \ge |P|$, there are 
			$\binom{\rev{i} - 1}{|P| - 1}$ \rev{$i$th-order} representative multi-indices 
			of $P$.
			\item \label{item:f}
			For all $\rev{i \in \mathbb Z_{> 0}}$, the common cumulant $f_{\rev{i}}: 
			2^{\rev{\monpaths}} \to \mathbb R$ satisfies \eqref{eq:f}.
			\item \label{item:mia-dist-estimation-proof}
			Statement \ref{item:mia-dist-estimation} of Theorem \ref{thm:mia-dist} is 
			true, i.e., Algorithm \ref{alg:mia-dist} correctly computes the common 
			cumulant vector for order $\rev{i} = n$.
		\end{enumerate} 
	\end{lemma}

	\subsection{Inversion Stage}
	\label{sect:inversion}
	
	In the inversion stage, we extract topological information from the vector of common 
	cumulants by applying an invertible linear transformation. Lemma 
	\ref{lem:mia-dist-estimation} \ref{item:f} shows that common cumulants are sums over 
	common link sets. But it is clear from \eqref{eq:L} and \eqref{eq:M} that common link 
	sets can be written as unions of exact link sets, which more directly provide 
	information about the routing matrix. Accordingly, common cumulants can be written as 
	sums over exact link sets, using \textit{exact cumulants}:
	
	\begin{definition}[Exact Cumulant] \label{def:g}
		\rev{For each positive integer $i$, we} define the \rev{$i$}th-order 
		\textit{exact cumulant} $g_{\rev{i}}: 2^{\rev{\monpaths}} \to \mathbb R$ by 
		\eqref{eq:g}, 
		\rev{replacing $n$ with $\rev{i}$}.
	\end{definition}
	
	\noindent
	In the following lemma, we formalize the relationship of common cumulants as sums of 
	exact cumulants. We then apply M\"obius inversion to this sum:
	
	\begin{lemma}[Properties of the Inversion Stage] \label{lem:mia-dist-inversion}
		Let $f_{\rev{i}}$ be the common cumulant vector, and let $g_{\rev{i}}: 
		2^{\rev{\monpaths}} \to \mathbb R$. The following three statements are equivalent:
		\begin{enumerate}
			\item \label{item:mia-dist-inversion-exact}
			$g_{\rev{i}}$ is the exact cumulant vector.
			\item \label{item:mia-dist-inversion-pre}
			$f_{\rev{i}}$ and $g_{\rev{i}}$ satisfy
			\begin{equation} \label{eq:pre-inversion}
			f_{\rev{i}}(P) = \sum_{Q \supseteq P} g_{\rev{i}}(Q), \qquad \forall \rev{P 
			\subseteq \monpaths}
			\end{equation}
			\item \label{item:mia-dist-inversion-post}
			$f_{\rev{i}}$ and $g_{\rev{i}}$ satisfy
			\begin{equation}
			g_{\rev{i}}(P) = \sum_{Q \supseteq P} (-1)^{|Q| - |P|} 
			f_{\rev{i}}(Q), \qquad \forall \rev{P \subseteq \monpaths} 
			\label{eq:inversion}
			\end{equation}
		\end{enumerate}
		Furthermore, statement \ref{item:mia-dist-inversion} of Theorem 
		\ref{thm:mia-dist} is true, i.e., the Algorithm \ref{alg:mia-dist} correctly 
		computes the exact cumulant vector.
	\end{lemma}

	\noindent
	Lemma \ref{lem:mia-dist-inversion} is the heart of MIA. By applying the inversion 
	\eqref{eq:inversion} to the vector of common 
	cumulants, we calculate the vector of \textit{exact} cumulants. Whereas common 
	cumulants contain information about which links are traversed by every path in a set, 
	exact cumulants contain information about which links are traversed 
	\textit{precisely} by the paths in a set, i.e., they contain information 
	about columns of the routing matrix.  
	
	\subsection{Reconstruction Stage}
	\label{sect:recon}
	
	The final stage of the algorithm is to reconstruct the routing matrix from the exact 
	cumulant vector. This reconstruction is straightforward, using only the 
	zero-nonzero pattern of $g_{\rev{i}}$:
	
	\begin{lemma}[Properties of the Reconstruction Stage] \label{lem:decode}
		Let $g_n: 2^{\rev{\monpaths}} \to \mathbb R$ be the exact cumulant vector. For 
		each $\rev{P \subseteq \monpaths}$, let $\rev{\chi(P, \monpaths)} \in \{0, 
		1\}^n$ be the characteristic vector of $P$ in $\rev{\monpaths}$. The following 
		are true:
		\begin{enumerate}
			\item
			If $P \in \supp(g_n)$, then 
			$\rev{\chi(P, \monpaths)}$ must be a column of the routing matrix. Under 
			Assumptions \ref{ass:1} and \ref{ass:2}, the converse is also true. 
			\item
			Statement \ref{item:mia-dist-recon} of Theorem \ref{thm:mia-dist} is true. 
		\end{enumerate}
	\end{lemma} 

%

	\subsection{Detailed Example}
	\label{sect:mia-dist-ex}
	
	In order to illustrate \rev{MIA}, we will apply the algorithm to a small example, 
	consisting of 3 monitor paths that utilize three links. We will walk through each of 
	the three stages of the algorithm in detail.
	 
	\paragraph*{Setup}
	Consider a network with three monitor paths $\rev{\monpaths} = \{p_1, p_2, p_3\}$ and 
	three links $\rev{L} = \{\ell_1, \ell_2, \ell_3\}$, with a routing matrix
	\begin{equation}
		\mathbf R = \begin{blockarray}{cccc}
		& \ell_1 & \ell_2 & \ell_3 \\
		\begin{block}{c(ccc)}
		p_1 & 1 & 1 & 0 \\
		p_2 & 1 & 0 & 1 \\
		p_3 & 0 & 0 & 1 \\
		\end{block}
		\end{blockarray}
		\label{eq:detailed-ex-R}
	\end{equation}
	Clearly this routing matrix satisfies Assumption \ref{ass:1}. Each of the three link 
	delay distributions is exponential, with probability density functions $f_{u_\ell}(x) 
	= \lambda_\ell e^{-\lambda_\ell x}$ \rev{for each $\ell \in L$}, and intensities 
	$\lambda_{\ell_1} = 1$, $\lambda_{\ell_2} = 1.5$, and 
	$\lambda_{\ell_3} = 2$ (in units of per millisecond). All cumulants of exponential 
	distributions are positive, so the latency variables satisfy Assumption \ref{ass:2}. 
	We then invoke \eqref{eq:v} to obtain the joint distribution of path delays. We 
	assume that the theoretical distribution of path delays is known---in particular, 
	the cumulants $\kappa_\alpha(\mathbf V)$ are known exactly---and our objective is to 
	use these cumulants to infer the routing matrix, via Algorithm \ref{alg:mia-dist}.
	
	\subsubsection{Estimation Stage}
	
	There are seven non-empty \rev{subsets of $\monpaths$}. Sets with one path only have 
	one 3rd-order representative 
	multi-index; for example, the path set $P = \{p_1\}$ has a unique representative 
	multi-index $\alpha = (3, 0, 0)$. Sets with two paths have 2 representative 
	multi-indices; for example, $P = \{p_1, p_2\}$ has $\alpha = (2, 1, 0)$ and 
	$\alpha' = (1, 2, 0)$. The three-element path set $P = \rev{\monpaths}$ has only the 
	one representative multi-index $\alpha = (1, 1, 1)$. For each of these seven 
	\rev{path sets}, we will select one of the representative 
	multi-indices arbitrarily and collect them into the common cumulant vector. For 
	example:
	\[
		\mathbf f_3 = \begin{pmatrix}
		f_3(\{p_1\}) \\ f_3(\{p_2\}) \\ f_3(\{p_3\}) \\
		f_3(\{p_1, p_2\}) \\ f_3(\{p_1, p_3\}) \\ f_3(\{p_2, p_3\}) \\
		f_3(\rev{\monpaths})
		\end{pmatrix} = 
		\begin{pmatrix}
			\kappa_{(3, 0, 0)}(\mathbf V) \\
			\kappa_{(0, 3, 0)}(\mathbf V) \\
			\kappa_{(0, 0, 3)}(\mathbf V) \\
			\kappa_{(1, 2, 0)}(\mathbf V) \\
			\kappa_{(1, 0, 2)}(\mathbf V) \\
			\kappa_{(0, 1, 2)}(\mathbf V) \\
			\kappa_{(1, 1, 1)}(\mathbf V)
		\end{pmatrix} = 
		\begin{pmatrix}
		70/27 \\
		9/4 \\
		1/4 \\
		2 \\
		0 \\
		1/4 \\
		0
		\end{pmatrix}
	\] 
	
	{ It is worth noting that $\mathbf f_3$ agrees with
	\eqref{eq:f}, i.e., we 
	can decompose the vector into univariate cumulants of link delays:
	\[
		\mathbf f_3 = \begin{pmatrix}
		\kappa_{(3, 0, 0)}(\mathbf V) \\
		\kappa_{(0, 3, 0)}(\mathbf V) \\
		\kappa_{(0, 0, 3)}(\mathbf V) \\
		\kappa_{(1, 2, 0)}(\mathbf V) \\
		\kappa_{(1, 0, 2)}(\mathbf V) \\
		\kappa_{(0, 1, 2)}(\mathbf V) \\
		\kappa_{(1, 1, 1)}(\mathbf V)
		\end{pmatrix}
		= \begin{pmatrix}
			\kappa_3(U_1) + \kappa_3(U_2) \\
			\kappa_3(U_1) + \kappa_3(U_3) \\
			\kappa_3(U_3) \\
			\kappa_3(U_1) \\
			0 \\
			\kappa_3(U_3) \\
			0
		\end{pmatrix} 
		= 
		\begin{pmatrix}
		70/27 \\
		9/4 \\
		1/4 \\
		2 \\
		0 \\
		1/4 \\
		0
		\end{pmatrix}
	\]
	Of course, performing this decomposition relies on our prior knowledge of $\mathbf R$ 
	and the 
	link delay distributions, which are unavailable to the experimenter.}
	
	\subsubsection{Inversion Stage} 
	
		In order to obtain the 
		exact cumulant vector $\mathbf g_3$ from the common cumulant vector $\mathbf 
		f_3$, we apply the 
		M\"obius inversion transformation \eqref{eq:inversion}. Note that this 
		transformation is linear, and it can be represented in the matrix form $\mathbf 
		g_3 = \mathbf X	\mathbf f_3$, where the matrix $\mathbf X$ contains the 
		coefficients $(-1)^{|Q| - |P|}$:
		
		\begingroup
		\setlength\arraycolsep{3pt}
		\small
		\[
		\begin{pmatrix} g_3(\{p_1\}) \\ g_3(\{p_2\}) \\ g_3(\{p_3\}) \\ 
		g_3(\{p_1, p_2\}) \\ g_3(\{p_1, p_3\}) \\ g_3(\{p_2, p_3\}) \\ 
		g_3(\rev{\monpaths}) \end{pmatrix}
		= \underbrace{\begin{pmatrix}
			1 & 0 & 0 & -1 & -1 & 0 & 1 \\
			0 & 1 & 0 & -1 & 0 & -1 & 1 \\
			0 & 0 & 1 & 0 & -1 & -1 & 1 \\
			0 & 0 & 0 & 1 & 0 & 0 & -1 \\
			0 & 0 & 0 & 0 & 1 & 0 & -1 \\
			0 & 0 & 0 & 0 & 0 & 1 & -1 \\
			0 & 0 & 0 & 0 & 0 & 0 & 1
			\end{pmatrix}}_{\mathbf X}
		\begin{pmatrix} f_3(\{p_1\}) \\ f_3(\{p_2\}) \\ f_3(\{p_3\}) \\ f_3(\{p_1, p_2\}) 
		\\ f_3(\{p_1, p_3\}) \\ f_3(\{p_2, p_3\}) \\ f_3(\rev{\monpaths}) \end{pmatrix}
		\]
		\normalsize
		\endgroup
		
		\noindent
		Evaluating this transformation, we obtain the following expression for the exact 
		cumulant vector:
		\[
		\mathbf g_3 =  \begin{pmatrix}
		g_3(\{p_1\}) \\ g_3(\{p_2\}) \\ g_3(\{p_3\}) \\
		g_3(\{p_1, p_2\}) \\ g_3(\{p_1, p_3\}) \\ g_3(\{p_2, p_3\}) \\
		g_3(\rev{\monpaths})
		\end{pmatrix} = 
		\begin{pmatrix}
		16/27 \\ 0 \\ 0 \\ 2 \\ 0 \\ 1/4 \\ 0
		\end{pmatrix} 
		\]
		We can verify that these values for $\mathbf g_3$ agree with both \eqref{eq:g} 
		and 
		\eqref{eq:pre-inversion}. For example, the routing matrix 
		\eqref{eq:detailed-ex-R} implies that $\rev{E}(\{p_1\}) = \{\ell_2\}$, so 
		\eqref{eq:g} gives
		\[
		g_3(\{p_1\}) 
		= \frac{2}{\lambda_{\ell_2}^3} 
		= \frac{16}{27}
		\]
		in agreement with our computed result for $\mathbf g_3$. Furthermore, 
		\eqref{eq:pre-inversion} claims that we can decompose $f_3(\{p_1\})$ according to
		\begin{align*}
		f_3(\{p_1\}) &= g_3(\{p_1\}) + g_3(\{p_1, p_2\}) + g_3(\{p_1, p_3\}) + 
		g_3(\rev{\monpaths}) \\
		&= \frac{70}{27}
		\end{align*}
		in agreement with $f_3(\{p_1\})$ obtained from the previous stage.
		
	\subsubsection{Reconstruction Stage}
	
		All that remains is to examine the zero-nonzero pattern of $\mathbf g_3$. Note 
		that $\mathbf g_3$ has three non-zero entries: $P_1 = \{p_1\}$, $P_2 = \{p_1, 
		p_2\}$, and $P_3 = \{p_2, p_3\}$. We can then reconstruct the routing matrix from 
		the characteristic vectors of these three path sets:
		\[
			\hat {\mathbf R} = \begin{pmatrix}
			\rev{\chi(P_1, \monpaths)} &
			\rev{\chi(P_2, \monpaths)} &
			\rev{\chi(P_3, \monpaths)} 
			\end{pmatrix} =
			\begin{pmatrix}
			1 & 1 & 0 \\
			0 & 1 & 1 \\
			0 & 0 & 1
			\end{pmatrix}
		\]
		Observe that $\hat{\mathbf R}$ is equivalent to the ``ground truth'' routing 
		matrix in \eqref{eq:detailed-ex-R}, modulo { an irrelevant} permutation of 
		columns, as 
		guaranteed by Theorem \ref{thm:mia-dist} \ref{item:mia-dist-recon}.

	\section{\rev{From Distributions to Data}}
	\label{sect:data}
	
	Having presented the core theory underlying MIA, we now turn to a more practical 
	problem: routing matrix inference from \textit{data}, rather than from a theoretical 
	distribution of path delays. \rev{Instead of knowing the joint distribution of the 
	path delay vector $\mathbf V$, in this section, we only assume that an i.i.d. sample 
	$\mathbf v_1, \mathbf v_2, \dots, \mathbf v_N \in \real^n$ of this distribution is 
	available. Thus, instead of using ground-truth cumulant values $\kappa_\alpha(\mathbf 
	V)$ in the estimation stage of the algorithm, we have to use estimates of these 
	cumulants via the $k$-statistics $k_\alpha(\mathbf v_1, \mathbf v_2, \dots, \mathbf 
	v_N)$. Moreover, because $k$-statistics} introduce noise into the inference 
	procedure, \rev{we will also need to modify the reconstruction stage to be robust 
	against this 
	noise.
	
	\paragraph*{Estimation Stage}
	
	In lines \ref{ln:midx} and \ref{ln:deff} of Algorithm \ref{alg:mia-dist}, MIA selects 
	an arbitrary representative multi-index $\alpha \in A_{n, P}$ and records the common 
	cumulant value $f_n(P) \leftarrow \kappa_\alpha(\mathbf V)$. The choice of 
	representative multi-index here is truly arbitrary, since all yield an identical 
	value for 
	$\kappa_\alpha(\mathbf V)$. This is not true for $k$-statistics. While the expected 
	values of $k_\alpha(\mathbf 
	v_1, \mathbf v_2, \dots, \mathbf v_N)$ are identical for all $\alpha \in A_{n, P}$, 
	the actual values of these statistics will generally be different. It is not clear 
	that any of these values is a better estimate than the others, so we propose 
	replacing $\kappa_\alpha(\mathbf V)$ with the simple average
	\begin{equation}
		\hat f_n(P) = \binom{n - 1}{|P| - 1}^{-1} \sum_{\alpha \in A_{n, P}} 
		k_\alpha(\mathbf v_1, \mathbf v_2, \dots, \mathbf v_N) \label{eq:f-hat}
	\end{equation}
	of all $k$-statistics for the representative multi-indices of $P$. Thus, we replace 
	both lines \ref{ln:midx} and \ref{ln:deff} in Algorithm \ref{alg:mia-dist} with 
	\eqref{eq:f-hat}, as well as using the notation $\hat f_n(P)$ instead of $f_n(P)$ (to 
	highlight that the algorithm is now using an estimate of the common cumulant instead 
	of its true value). 
        
	\paragraph*{Inversion Stage}
	
	There is no need to modify the inversion stage of the algorithm in the data-driven 
	setting. The inversion stage simply applies the linear transformation $\mathbf g_n = 
	\mathbf X \mathbf f_n$, where $\mathbf X$ encodes the M\"obius inversion. When we 
	switch from $\mathbf g_n$ and $\mathbf f_n$ to vectors of estimates $\hat{\mathbf 
	g}_n$ and $\hat{\mathbf f}_n$, this transformation is still valid in expectation:
	\[
		\E[\hat{\mathbf g}_n] = \mathbf X \E[\hat{\mathbf f}_n] = \mathbf X \mathbf f_n
		= \mathbf g_n
	\]
	
	\paragraph*{Reconstruction Stage}
	
	In line \ref{ln:gnnz} of Algorithm \ref{alg:mia-dist}, MIA checks if an entry of the 
	exact cumulant vector is nonzero. But in the data-driven scenario, we switch from 
	exact cumulants to estimates $\hat{\mathbf g}_n$, which only match the zero-nonzero 
	pattern of $\mathbf g_n$ in expectation. To account for inevitable noise in these 
	estimates, instead of checking if $\hat g_n(P) = 0$, we must adopt some kind of 
	hypothesis test $\texttt{Nonzero}(g_n(P) \mid \mathbf v_1, \mathbf v_2, \dots, 
	\mathbf v_N)$, i.e., some decision rule to guess whether $g_n(P) \ne 0$ based on the 
	data. We will examine the construction of such a test in the next subsection.
	}
 
	The performance of \rev{MIA in the data-driven setting} depends 
	entirely on the accuracy of the hypothesis test. This accuracy depends on 
	the test itself, the choice of test parameters (like significance levels), and the 
	size of the sample size $N$, so it is difficult to state general theoretical 
	guarantees regarding the algorithm. Nonetheless, some guarantees are evident in 
	extreme cases, if Assumptions \ref{ass:1} and \ref{ass:2} are satisfied:
	\begin{enumerate}
		\item
		If the test has no Type I error, i.e., if $g_n(P) = 0$ always leads to a decision 
		that $\rev{\texttt{Nonzero}}(g_n(P) \mid \mathbf v_1, \mathbf v_2, \dots, \mathbf 
		v_N)$ is false, then every column of $\hat{\mathbf R}$ will be a true column of 
		$\mathbf R$.
		\item
		If the test has no Type II error, then $\hat{\mathbf R}$ will contain every 
		column of $\mathbf R$.
		\item
		If the test is \textit{consistent}, in the sense that the test is free of both 
		Type I and Type II error in $N \to \infty$ limit, then similarly $\hat{\mathbf R} 
		= \mathbf R$ in 
		the $N \to \infty$ limit.
	\end{enumerate}

	\noindent
	For all practical purposes, none of these extreme cases will apply, and we will have 
	to rely on the algorithm's performance in test scenarios to assess its usefulness.
	
	\subsection{Hypothesis Tests}
	\label{sect:hypothesis-testing}
	
	We now examine the hypothesis test $\rev{\texttt{Nonzero}}(g_n(P) 
	\mid \mathbf v_1, \mathbf v_2, \dots, \mathbf v_N)$, \rev{which we will 
	subsequently abbreviate as $\texttt{Nonzero}(g_n(P))$. Because $\E[\hat g_n(P)] = 
	g_n(P)$,} we can assess the null hypothesis $g_n(P) = 0$ via an equivalent null 
	hypothesis, that $\E[\rev{\hat g_n(P)}] = 0$. There is no single correct 
	way to perform this \rev{mean location} test---many approaches exist, with advantages 
	and disadvantages. 
        
%
	
	\subsubsection{Normal Approximation}
	
	Because \rev{the statistics $\hat g_n(P)$} are asymptotically normally distributed, 
	we could simply estimate the mean and variance of the distribution and apply a 
	standard $z$-test. 
	This approach is used in \cite{BS-SR-SG:10}, for example, to perform hypothesis 
	testing on univariate cumulants, using univariate $k$-statistics. Unfortunately, 
	while the mean of the distribution is easily estimated by $\rev{\hat g_n(P)}$, the 
	variance relies on computing variances of multivariate $k$-statistics, 
	which are both mathematically and computationally complex. 

	\subsubsection{Sample Splitting}
	
	Another simple approach is to partition the original $N$-length sample into $M$ 
	subsamples of size $N / M$, compute \rev{$\hat g_n(P)$} for each subsample, and use 
	standard hypothesis testing to assess whether the statistics have zero mean. Since 
	the subsamples are non-overlapping, each of the $M$ values of \rev{$\hat g_n(P)$} 
	will be iid, so standard approaches (like the 1-sample Student's $t$-test 
	\cite[\S 9.5]{MHD-MJS:12}) can be used to test the null hypothesis that
	$\E[\rev{\hat g_n(P)}] = 0$.  
	
	\subsubsection{Bootstrapping}
	
	Bootstrapping (see, e.g., \cite[Chapter 2]{ACD-DVH:99}) is a resampling 
	technique that uses the empirical distribution (i.e., the discrete distribution with 
	uniform weight on each sample value) to approximate the original distribution. For $b 
	= 1, 2, \dots, M$ (where typically $M \approx 50$), we define a \textit{resample} 
	$\tilde {\mathbf v}_{b1}, \tilde {\mathbf v}_{b2}, \dots, \tilde {\mathbf v}_{bN}$ 
	that is chosen randomly with replacement from 
	the original sample $\mathbf v_1, \mathbf v_2, \dots, \mathbf v_N$. We then compute 
	\rev{$\hat g_n(P)$} for each resample, resulting in a sample of \rev{size $M$ for 
	$\hat g_n(P)$}, which we can use to perform a mean hypothesis test. This approach has 
	been applied to estimating confidence intervals for cumulants \cite{YZ-DH-ANV:93}. 
	
	\subsection{Detailed Example}
	\label{sect:ex-data}
	

	In order to illustrate the empirical version of MIA, we 
	will continue to use the low-dimensional example from Section \ref{sect:mia-dist-ex}, 
	with the same routing matrix \eqref{eq:detailed-ex-R} and the same 
	exponentially-distributed link delays. \rev{We created a synthetic dataset with} 900 
	independent samples from each link distribution, \rev{which we transformed into 900} 
	samples of $V_{p_1}$, \rev{$V_{p_2}$ , and $V_{p_3}$} based on the sums encoded in 
	the routing matrix.  
	
	We use the sample splitting approach to the \rev{$\texttt{Nonzero}(g(P))$
	hypothesis} test in this example. The 900 original sample points are split into 30 
	samples of size 30. 
	\rev{To carry out the estimation stage, we estimate the common cumulant vector for 
	each of these 30 samples with the simple average of $k$-statistics in 
	\eqref{eq:f-hat}:
	}
	\[
		\hat {\mathbf f}_3 = \begin{pmatrix}
			\hat f_{3}(\{p_1\}) \\
			\hat f_{3}( \{p_2\}) \\
			\hat f_{3}(\{p_3\}) \\
			\hat f_{3}(\{p_1, p_2\}) \\
			\hat f_{3}(\{p_1, p_3\}) \\
			\hat f_{3}(\{p_2, p_3\}) \\
			\hat f_{3}(\rev{\monpaths})
		\end{pmatrix} = 
		\begin{pmatrix}
			k_{(3, 0, 0)}(\cdot) \\
			k_{(0, 3, 0)}(\cdot) \\
			k_{(0, 0, 3)}(\cdot) \\
			\frac 1 2 k_{(1, 2, 0)}(\cdot) + \frac 1 2 k_{(2, 1, 0)}(\cdot) \\
			\frac 1 2 k_{(1, 0, 2)}(\cdot) + \frac 1 2 k_{(2, 0, 1)}(\cdot) \\
			\frac 1 2 k_{(0, 1, 2)}(\cdot) + \frac 1 2 k_{(0, 2, 1)}(\cdot) \\
			k_{(1, 1, 1)}(\cdot)
		\end{pmatrix}
	\]
	\rev{Here $k_\alpha(\cdot)$ is shorthand for $k_\alpha(\mathbf v_1, \mathbf v_2, 
	\dots, \mathbf v_N)$. Columns 2 and 3 of} Table \ref{table:small-results} report the 
	means and standard errors for \rev{these 30 estimates of $\hat{\mathbf f_3}$. To 
	perform the inversion stage, the the vector $\hat{\mathbf g}_3$ is then computed} by 
	$\hat{\mathbf g}_3 = \mathbf X \hat{\mathbf{f}}_3$, where $\mathbf X$ is the matrix 
	defined in Section \ref{sect:mia-dist-ex}. \rev{Columns 4 and 5 of Table 
	\ref{table:small-results} similarly summarize the distribution of these 30 estimates 
	for $\hat{\mathbf g_3}$.} Indeed, all of the $\rev{\hat f_3(P)}$ and \rev{$\hat 
	g_3(P)$ averages} are within one standard error of $f_3(P)$ and $g_3(P)$, 
	respectively.

	\begin{table}
		\centering
		\begin{tabular}{r|rr|rr} $P$ & $f_3(P)$ & $\hat f_{3, P}$ & $g_3(P)$ & $\hat g_{3, P}$ \\ 
\hline$\{p_1\}$ & $2.59$ & $2.67 \pm 0.5$ & $0.593$ & $0.66 \pm 0.2$ \\$\{p_2\}$ & $2.25$ 
& $2.31 \pm 0.7$ & $0$ & $0.06 \pm 0.2$ \\$\{p_3\}$ & $0.25$ & $0.24 \pm 0.05$ & $0$ & 
$0.02 \pm 0.02$ \\$\{p_1,p_2\}$ & $2$ & $2.01 \pm 0.6$ & $2$ & $2.01 \pm 0.5$ 
\\$\{p_1,p_3\}$ & $0$ & $-0.01 \pm 0.05$ & $0$ & $-0.01 \pm 0.04$ \\$\{p_2,p_3\}$ & 
$0.25$ & $0.23 \pm 0.07$ & $0.25$ & $0.23 \pm 0.06$ \\$\{p_1,p_2,p_3\}$ & $0$ & $0.00 \pm 
0.09$ & $0$ & $0.00 \pm 0.09$ \end{tabular}
		\caption{Common and exact cumulants in the low-dimensional 
			example. Columns $f_3(P)$ and $g_3(P)$ report the true underlying values, 
			while $\hat f_3(P)$ and $\hat g_3(P)$ show the mean and standard 
			error of the respective estimates.}
		\label{table:small-results}
	\end{table} 

	\rev{Based on these 30 estimates of $\hat{\mathbf g}_3$, we perform the 
	reconstruction stage using} a 1-sample Student's $t$-test to assess the null 
	hypothesis that $\E[\rev{\hat g_3(P)}] = 0$ for each path set. The $p$-value for each 
	null hypothesis is reported in Table \ref{table:small-testing}, as well as the result 
	of the test with a significance threshold of 0.01.
	
	\begin{table}
 		\centering
 		\begin{tabular}{r|rl} $P$ & p-value for $g_3(P) = 0$ & $\chi(P)$ is in $R$? \\ \hline$\{p_1\}$ & $0.001$ & Yes \\$\{p_2\}$ & $0.8$ & No \\$\{p_3\}$ & $0.5$ & No \\$\{p_1,p_2\}$ & $0.0005$ & Yes \\$\{p_1,p_3\}$ & $0.9$ & No \\$\{p_2,p_3\}$ & $0.0008$ & Yes \\$\{p_1,p_2,p_3\}$ & $1$ & No \end{tabular}
 		\caption{Hypothesis testing for whether or not $\rev{\chi(P, \monpaths)}$ is a 
 		column of the routing matrix, at 0.01 significance.}
 		\label{table:small-testing}
 	\end{table}
  	
  	For precisely three of the path sets, we reject the null hypothesis that $g_3(P) = 
  	0$: $P_1 = \{p_1\}$, $P_2 = \{p_1, p_2\}$, and $P_3 = \{p_2, p_3\}$. Assembling the 
 	characteristic vectors of these path sets into $\hat{\mathbf R}$, we obtain \rev{an 
 	identical estimate} to our result from Section \ref{sect:mia-dist-ex}, \rev{which} is 
 	identical to the ground truth routing matrix (up to a permutation of columns). 

\rev{
	\section{Sparse M\"obius Inference} 
	\label{sect:sparse}
	
	The key step in the M\"obius Inference Algorithm is the linear transformation 
	$\mathbf g_i = \mathbf X \mathbf f_i$, where $\mathbf g_i$ is a vector of $2^n - 1$ 
	exact cumulants, $\mathbf f_i$ is a vector of $2^n - 1$ common cumulants, $n$ is the 
	number of monitor paths, and $\mathbf X$ is the matrix encoding M\"obius inversion. 
	Three problems arise naturally: 
	the computational expense of the transformation $\mathbf X$, the
	impracticality of populating every entry of $\mathbf f_i$ with empirical 
	measurements, 
	and the noise present in $\mathbf f_i$ (and $\mathbf g_i$) due to the use of 
	cumulants with excessively high order. In this section, we simultaneously tackle  these three 
        problems using several different  sparsity heuristics. 
	
	Our proposed ``Sparse M\"obius Inference'' procedure proceeds in three stages. In the 
	first stage, we use measurements of low-order common cumulants to identify which 
	entries of the $\mathbf f_i$ and $\mathbf g_i$ vectors can contain nonzero entries. We 
	can 
	then ignore all other entries of these vectors and drop their corresponding columns and 
	rows from $\mathbf X$, reducing the M\"obius inversion down to a (typically much) 
	smaller set of equations. In the second stage, we impose the following sparsity 
	heuristic on $\mathbf g_i$: if $P$ is a sufficiently large path set that is strictly 
	contained within some other path set in $\supp(f_i)$, then $g_i(P) = 0$. This 
	heuristic allows us to remove further entries from both $\mathbf g_i$ and $\mathbf 
	f_i$, 
	provided  we make a suitable modification to $\mathbf X$. Finally, in the third 
	stage, we apply a sparsity-promoting lasso optimization problem to filter noisy 
	estimates of common cumulants and impute the values of common cumulants that are 
	impractical to measure. The end result is a sparse estimate for $\mathbf g_i$, which 
	only relies on estimates of common cumulants up to a small, user-specified order.
	
	\subsection{Stage 1: Bound the Support of $f_i$}
	
	In the first stage, we estimate the collection of path sets $P \subseteq \monpaths$ 
	for which $f_i(P) \ne 0$. The key to this process is the observation that $f_i(Q) \ne 
	0$ only if $f_i(P) \ne 0$ for all subsets $P \subseteq Q$: if just a single subset 
	$P$ has a zero-valued common cumulant, then $C(P) = \emptyset$, which implies that 
	$C(Q) = \emptyset$. If we focus on small path sets, then we can use low-order 
	cumulants to identify which of these path sets have no common links, and remove all 
	of their supersets from the support of $\mathbf f_i$. 
	
	We can maintain a compact representation of our estimate of $\supp(\mathbf f_i)$ 
	using a \textit{bounding topology}. A bounding topology is any collection of path 
	sets 
	$\mathcal B \subseteq 2^{\monpaths}$ with the following property: if $f_i(P) \ne 0$, 
	then $\mathcal B$ contains some path set $B \in \mathcal B$ such that $P \subseteq 
	B$. We will refer to the collection of all sets contained by some $B \in \mathcal B$ 
	(i.e., the union $\bigcup_{B \in \mathcal B} 2^B$) as the ``support estimate'' of 
	$\mathcal B$. Below are two extreme examples:
	\begin{itemize}
		\item $\mathcal B = \{\monpaths\}$ is trivially a bounding 
		topology, albeit not a very informative one, since the support estimate is 
		$2^\monpaths$.
		\item $\mathcal B = \supp(\mathbf g_i)$ is a bounding topology: if $f_i(P) \ne 
		0$, then 
		some superset $B \supseteq P$ satisfies $g_i(B) \ne 0$, and thus $B \in \mathcal 
		B$. This is a ``tight'' bounding topology, in the sense that every set in its 
		support estimate is indeed in the support of $\mathbf f_i$.
	\end{itemize} 
	Stage 1 begins with an uninformative bounding topology (like $\mathcal B = 
	\{\monpaths\}$), 
	and it iteratively ``tightens'' $\mathcal B$ using successive orders of common 
	cumulant estimates. The fundamental idea is that if we determine 
	$\texttt{Nonzero}(f_i(P))$ is false for some small path set $P$, then we ought to 
	split up all $B \in \mathcal B$ containing $P$ into smaller sets that do not contain 
	$P$, thereby eliminating all supersets of $P$ from the support estimate. This 
	iterative tightening procedure then terminates at a (typically small) 
	user-specified cumulant order. 
	
	Unfortunately, $\texttt{Nonzero}(f_i(P))$ is usually a hypothesis test with limited 
	statistical power---there is a chance that our data would incorrectly indicate that 
	$f_i(P) = 0$, leading us to remove any superset of $P$ from the support estimate and 
	thus ignore nonzero values of the common cumulant in future calculations. Such an 
	error could greatly harm the accuracy of later stages of the topology inference. In 
	order to hedge against this possibility, we propose a robust procedure that splits a 
	set $B \in \mathcal B$ only if a sufficient number of subsets of $B$ are found to have
	zero common cumulant. The user provides a \textit{threshold function} $t: \mathbb 
	Z_{> 0} \times \mathbb Z_{> 0} \to \mathbb Z_{> 0}$, where $B \in \mathcal B$ is 
	never split so long as
	$t(|B|, i)$ size-$i$ subsets of $|B|$ are found to have a nonzero common cumulant.  
	
	The core of the procedure is Algorithm \ref{alg:refine}, which tightens an estimate 
	of the bounding topology using common cumulants of some fixed order $i$. The 
	algorithm initially computes the collection of all size-$i$ sets $P$ in the support 
	estimate of $\mathcal B$ for which $\texttt{Nonzero}(f_i(P))$ is true. What follows 
	is effectively a voting procedure: each of these sets $P$ counts as a ``vote'' in 
	favor of keeping each superset $Q \supseteq P$ in the support estimate. If one of the 
	sets $B \in \mathcal B$ fails to reach its threshold of $t(|B|, i)$ votes, then 
	$B$ is split up into the $|B|$ subsets obtained by removing one element from $B$, and 
	the votes for these subsets are tallied as well. This process repeats until all 
	the sets in $\mathcal B$ with size at least $i$ reach their respective thresholds. 
	Theorem \ref{thm:refine} formally states the guarantees of this algorithm:
	
	\begin{algorithm}
		\caption{\rev{$\texttt{Tighten}(\mathcal B, i, t)$}}
		\begin{algorithmic}[1] \rev{
			\REQUIRE Bounding topology $\mathcal B \subseteq 2^{\monpaths}$, cumulant 
			order $i \in \mathbb{Z}_{> 0}$, and threshold function $t: \mathbb{Z}_{> 0} 
			\times \mathbb{Z}_{> 0} \to \mathbb{Z}_{> 0}$
			\ENSURE Tightened bounding topology $\mathcal B' \subseteq 2^{\monpaths}$
			\STATE Initialize $\mathcal B' = \emptyset$, $\mathcal X = \emptyset$, and
			\[
			\mathcal P = \left\{
			P \in \bigcup_{B \in \mathcal B} \binom{B}{i} : 
			\texttt{Nonzero}(f_i(P))
			\right\}
			\]
			\WHILE{$|\mathcal B| > 0$}
			\STATE Remove an arbitrary set $B$ from $\mathcal B$ and add it to 
			$\mathcal X$
			\IF{$|B| < i$ or $|\{P \in \mathcal P : P \subseteq B\}| \ge t(|B|, i)$} 
			\label{ln:count}
			\STATE $\mathcal B' \leftarrow \mathcal B' \cup \{B\}$ \label{ln:b-add}
			\ELSE
			\FOR{$p \in B$}
			\STATE $B_{\rm sub} \leftarrow B \setminus \{p\}$
			\IF{$B_{\rm sub} \notin \mathcal X$ and no set in $\mathcal B \cup \mathcal 
				B'$ contains $B_{\rm sub}$} \label{ln:if}
			\STATE $\mathcal B \leftarrow \mathcal B \cup \{B_{\rm sub}\}$
			\ENDIF
			\ENDFOR
			\ENDIF
			\ENDWHILE
			\RETURN $\mathcal B'$}
		\end{algorithmic}
		\label{alg:refine}
	\end{algorithm}
	
	\begin{theorem}[Properties of Algorithm \ref{alg:refine}] \label{thm:refine}
		Let $\mathcal B \subseteq 2^{\monpaths}$ be a collection of path sets, let $i \in 
		\mathbb{Z}_{> 0}$ be a cumulant order, and let $t: \mathbb{Z}_{> 0} \times 
		\mathbb{Z}_{> 0} \to \mathbb{Z}_{> 0}$ be a 
		threshold function. The following are true:
		\begin{enumerate}
			\item Algorithm \ref{alg:refine} evaluates $\texttt{IsNonzero}(f_i(P))$ 
			$O(n^i)$ times and terminates after $O(2^q)$ iterations of the while loop, 
			where $q$ is the size of the largest set in $\mathcal B$. The algorithm 
			returns a collection of path sets $\mathcal B' \subseteq 2^{\monpaths}$. 
			\item The support estimate of $\mathcal B'$ is a subset of the support 
			estimate of $\mathcal B$. 
			\item For any set $P$ in the support estimate of $\mathcal B$, $P$ is also in 
			the support estimate of $\mathcal B$ if either $|P| < i$, or if there is a 
			superset $Q \supseteq P$ in the support estimate of $\mathcal B$ for which at 
			least $t(|Q|, i)$ size-$i$ subsets $R \subseteq Q$ satisfy 
			$\texttt{Nonzero}(f_i(R))$. 
		\end{enumerate}
	\end{theorem}
	
	\begin{proof}
		There are at most $\binom{n}{i} = O(n^i)$ size-$i$ sets, so 
		$\texttt{Nonzero}(f_i(P))$ is evaluated $O(n^i)$ times to compute $\mathcal P$. 
		The worst-case runtime occurs when $|\{P \in \mathcal P : P \subseteq B\}| < 
		t(|B|, i)$ for each iteration of the while loop, in which case the variable $B$ 
		takes on the value of every subset (with size at least $i$) of every original set 
		in $\mathcal B$ precisely once (because the collection $\mathcal X$ tracks which 
		sets have already been processed, preventing redundant iterations of the while 
		loop). Thus, there are $O(2^q)$ iterations of the while loop. 
		
		To prove (ii), observe that every set added to $\mathcal B'$ was originally in 
		the queue $\mathcal B$, and that sets in the queue are either from the original 
		collection $\mathcal B$, or they are subsets of a previous element in the queue. 
		Hence every set in $\mathcal B'$ is a subset of a set in the original $\mathcal 
		B$, so the support estimate of $\mathcal B'$ is a subset of the original support 
		estimate. To prove (iii), suppose that $P$ is in the support estimate of 
		$\mathcal B'$, so that some $B' \in \mathcal B'$ contains $P$. Sets are only 
		added to $\mathcal B'$ on line \ref{ln:b-add}, and the set must satisfy either 
		$|B'| < i$ or $|\{P' \in \mathcal P : P' \subseteq B'\}| \ge t(|B'|, i)$, i.e., 
		(b) is satisfied with $Q = B'$.  
	\end{proof}
	
	Through the repeated application of Algorithm \ref{alg:refine} to a collection 
	$\mathcal B$ and successively larger orders $i$, as detailed in Algorithm 
	\ref{alg:bound}, we obtain tighter
	support estimates. Every path set in $\supp(\mathbf f_i)$ should remain in the 
	support 
	estimate of $\mathcal B$ after 
	each iteration, so long as the values of the threshold function $t$ are sufficiently 
	small (and the test $\texttt{Nonzero}(f_i(P))$ is sufficiently accurate). 
	Furthermore, as we incorporate information from higher-order cumulants, we 
	remove path sets for which $f_i(P) = 0$ from the support estimate. In summary, the 
	support estimate of $\mathcal B$ becomes a more and more accurate approximation of 
	$\supp(\mathbf f_i)$.
	
	\begin{algorithm}
		\caption{\rev{$\texttt{BoundingTopology}(\mathcal B, i_0, i_f, t)$}} 
		\label{alg:bound}
		\begin{algorithmic}[1] \rev{
			\REQUIRE Initial guess $\mathcal B \subseteq 2^{\monpaths}$, 
			initial cumulant order $i_0$, final cumulant order $i_f$, and threshold 
			function $t: \mathbb{Z}_{> 0} \times \mathbb{Z}_{> 0} \to \mathbb{Z}_{> 0}$
			\ENSURE Tightened bounding topology $\mathcal B \subseteq 2^{\monpaths}$
			\FOR{$i = i_0, i_0 + 1, \dots, i_f$}
			\STATE $\mathcal B \leftarrow \texttt{Tighten}(\mathcal B, i, t)$
			\ENDFOR
			\RETURN $\mathcal B$}
		\end{algorithmic}
	\end{algorithm}
	
	We will conclude the discussion of Stage 1 by addressing two questions---how should 
	we select the initial guess for $\mathcal B$ that is supplied to Algorithm 
	\ref{alg:bound}, and how should we design the threshold function $t$?
	
	\paragraph*{Choosing an Initial Bounding Topology}
	
	A safe (albeit inefficient) choice for the initial guess of bounding topology is 
	$\mathcal B = \{2^{\monpaths}\}$. Clearly the support estimate of $\mathcal B$ will 
	contain every path set in $\supp(\mathbf f_i)$. Unfortunately, this choice also 
	maximizes the 
	runtime of Algorithm \ref{alg:bound}, since the sub-routine Algorithm 
	\ref{alg:refine} is exponential in the size of the largest set in $\mathcal B$. 
	
	A more practical approach is to use second-order cumulants (i.e., covariances) to 
	construct an initial guess for $\mathcal B$. Second-order $k$-statistics tend to have 
	a small variance (compared to the higher-order $k$-statistics), leading to only a 
	small probability that $\texttt{Nonzero}(f_2(P))$ yields a false negative, which 
	makes the thresholding in Algorithm \ref{alg:refine} unnecessary. If we require that 
	$\texttt{Nonzero}(f_2(P))$ is true for \textit{all} two-element subsets of each set 
	in $\mathcal B$, then we can use second-order cumulants to construct a more efficient 
	initial guess for $\mathcal B$, and then we can run Algorithm \ref{alg:bound} on this 
	initial guess starting at order $i_0 = 3$. 
	
	One way to efficiently construct this covariance-based initial guess is to use 
	standard algorithms for maximal clique enumeration. Recall from graph theory that a 
	\textit{clique} is any set of nodes for which all nodes in the set are adjacent, and 
	a \textit{maximal clique} is a clique that is not contained within a larger clique. 
	Construct a graph $G_b = (\monpaths, E_b)$ where each monitor path is a node, 
	and an edge $\{p_i, p_j\}$ is included in $E_b$ if and only if 
	$\texttt{Nonzero}(f_2(\{p_i, p_j\}))$ is true. Cliques in $G_b$ are precisely the 
	path sets for which $\texttt{Nonzero}(f_2(P))$ is true of every two-element subset. 
	Therefore, we take as our initial guess for $\mathcal B$ the set of maximal cliques 
	in $G_b$. The size of the largest clique is typically significantly smaller than $n$, 
	leading to a faster runtime for Algorithm \ref{alg:bound}. 
	
	\paragraph*{Constructing the Threshold Function}
	
	Algorithm \ref{alg:bound} requires the user to specify a threshold function $t(|P|, 
	i)$, indicating the minimum number of size-$i$ subsets of $P$ that must pass the 
	nonzero common cumulant test for $P$ to remain in the support estimate. Choosing the 
	threshold value is a balance---large values may lead to sets in $\supp(f_i)$ being 
	rejected from the support estimate, but small values will cause information from many 
	zero-valued cumulants to be ignored. We will try to devise an intuitive and tunable 
	form for $t(|P|, i)$ to strike this balance. 
	
	Recall that the \textit{statistical power} of a hypothesis test is the probability of 
	rejecting the null hypothesis given that the alternative hypothesis is true---in our 
	case, the probability that $\texttt{Nonzero}(f_i(P))$ is true if indeed $P \in 
	\supp(\mathbf f_i)$. Suppose that, for each $P \in \supp(\mathbf f_i)$, the 
	corresponding test 
	$\texttt{Nonzero}(f_i(P))$ is true independently and with uniform probability $1 - 
	\beta$. Under these (inaccurate but nonetheless useful) assumptions, the number of 
	size-$i$ subsets of any $Q \in \supp(\mathbf f_i)$ for which 
	$\texttt{Nonzero}(f_i(P))$ is 
	true follows a binomial distribution, with $\binom{|Q|}{i}$ trials and a success 
	probability of $1 - \beta$. Hence, the probability that at least $t(|Q|, i)$ size-$i$ 
	subsets of $Q$ pass the nonzero test is $1 - F_{|Q|, i}(t(|Q|, i))$, where $F_{|Q|, 
	i}$ is the cdf of the binomial distribution. 
	
	Because $Q$ truly belongs to the support of $\mathbf f_i$, it is highly undesirable 
	that we 
	erroneously remove $Q$ from the support estimate by setting the threshold $t(|Q|, i)$ 
	inappropriately high. To render such an error unlikely, we must ensure that $1 - 
	F_{|Q|, i}(t(|Q|, i))$ exceeds some high probability $1 - \gamma \in (0, 1)$, e.g., 
	$1 - \gamma = 0.1$. Once we specify $\gamma$, we can solve for the appropriate 
	threshold as the quantity
	\begin{align*}
	t(|Q|, i) &= \max\{t \in \mathbb{Z}_{> 0} : F_{|Q|, i}(t) < \gamma\} \\
	&= \min \{t \in \mathbb{Z}_{> 0} : F_{|Q|, i}(t) \ge \gamma \} - 1
	\end{align*}
	In other words, we set $t(|Q|, i)$ as one less the $\gamma$ quantile of the 
	binomial distribution with $\binom{|Q|}{i}$ trials and success probability $1 - 
	\beta$. There is no good closed-form expression for the value of this quantile; 
	however, it is readily computable in many statistics packages.
	
	This binomial quantile specification for $t(|Q|, i)$ is somewhat informal, since the 
	outcomes of $\texttt{Nonzero}(f_i(P))$ are neither independently nor identically 
	distributed, as the derivation assumed. However, the method does at least provide an 
	intuitive way to reduce the specification of $t$ down to two tunable parameters, 
	$\gamma \in (0, 1)$ (the highest tolerable probability that $Q \in \supp(\mathbf 
	f_i)$ is accidentally rejected) and $\beta \in (0, 1)$ (an estimate for the 
	probability that $\texttt{Nonzero}(f_i(P))$ yields a false negative). We could also 
	specify different values of these parameters for different $k$-statistic orders $i$, 
	to account for the fact that $k$-statistics tend to become less accurate with higher 
	orders. 
	
	\subsection{Stage 2: Bound the Support of $g_i$} 
	
	In the previous stage, we used information from low-order cumulants to narrow the 
	entries of $\mathbf f_i$ containing nonzero entries down to the support estimate of 
	$\mathcal B$. Because $f_i(P) = 0$ implies that $g_i(P) = 0$ as well, this stage also 
	simultaneously restricts the nonzero entries of $\mathbf g$ to to the support estimate 
	of $\mathcal B$. The second stage drops even more zero-valued entries from 
	these two vectors. Instead of using empirical information from low-order cumulants, 
	this stage enforces a ``hard'' sparsity heuristic: that $g_i(P) = 0$ for all path 
	sets $P$ larger than some threshold size $s$, unless that path set is an element of 
	$\mathcal B$. In other words, we assume that the only ``large'' path sets are those 
	contained directly in the bounding topology inferred from low-order cumulants. 
	
	This heuristic immediately zeros out large swaths of the $\mathbf g_i$ vector, 
	allowing 
	us to ignore them during the final stage. But the heuristic also allows us to drop 
	even more entries from the $\mathbf f_i$ vector, as stated in the following lemma: 
	
	\begin{lemma}[Elimination of Large, Non-Maximal Path Sets]
		Let $\mathcal B \subseteq 2^{\monpaths}$ be a collection of path sets, and let $s 
		\in \mathbb{Z}_{> 0}$. Assume that the following are true:
		\begin{enumerate}
			\item Every set in $\mathcal B$ is maximal (i.e., no $B, B' \in \mathcal B$ 
			exist such that $B \subset B'$), \label{cond:maximal}
			\item $f_i(P) \ne 0$ and $g_i(P) \ne 0$ only if $P$ is in the support 
			estimate of $\mathcal B$, and
			\item $g_i(P) = 0$ for all $P \subseteq \monpaths$ with $|P| > s$ and $P 
			\notin \mathcal B$. 
			\label{cond:heuristic}
		\end{enumerate}
		Then for every $P$ in the support estimate of $\mathcal B$ such that $|P| \le s$, 
		\begin{align}
		\begin{split}
		g_i(P) &= \sum_{Q \supseteq P : |Q| \le s} (-1)^{|Q| - |P|} f_i(Q) \\
		&\qquad - \sum_{B \in \mathcal B : B \supseteq P} (-1)^{s - |P|} \binom{|B| - |P| 
		- 1}{s - |P|} f_i(B)  
		\end{split} \label{eq:mobius-mod}
		\end{align}
	\end{lemma}
	
	
	Due to \eqref{eq:mobius-mod}, there is no need to measure or keep track of $f_i(P)$ 
	for sufficiently large $P$, unless $P$ is a set in $\mathcal B$. Note that these 
	common cumulants are not just zeroed out---they take on a nonzero value; however, 
	this value is constrained to a linear combination of the common cumulants for $B \in 
	\mathcal B$, which are already elements of the common cumulant vector. 
	
	\subsection{Stage 3: Lasso Optimization}
	
	The previous two stages eliminated large parts of the $\mathbf f_i$ and $\mathbf g_i$ 
	vectors, using a combination of information from low-order cumulants, \textit{a 
		priori} assumptions, and suitable modifications of the M\"obius transformation 
		matrix 
	$\mathbf X$. These two stages significantly reduce the computational expense of 
	performing M\"obius inversion and populating $\mathbf f_i$ with empirical estimates 
	of 
	common cumulants. Furthermore, because the first stage tends to eliminate the largest 
	subsets of $\monpaths$ from the support for $\mathbf f_i$, we can populate $\mathbf 
	f_i$ 
	with cumulants of order lower than $n$. But this cumulant order (which must be at 
	least the size of the largest path set with a nonzero common cumulant) can still be 
	unrealistically large, and the resulting common cumulant estimates can be quite 
	noisy. In the final stage of Sparse M\"obius Inference, we address these two problems 
	by filtering $\mathbf f_i$ using lasso optimization.
	
	To set up the problem, the user first supplies a maximum cumulant order $i_{\rm max} 
	\in \mathbb{Z}_{> 0}$, indicating the largest order of cumulant they are willing to 
	estimate. Based on $i_{\rm max}$, we partition the common cumulant vector by $\mathbf 
	f_{i_{\rm max}} = \begin{pmatrix} \mathbf f_o & \mathbf f_u \end{pmatrix}^\top$, and 
	we make the 
	corresponding partition to the inversion matrix $\mathbf X = \begin{pmatrix} \mathbf 
	X_o & \mathbf X_u \end{pmatrix}$. $\mathbf f_o$ 
	corresponds to the common cumulants $f_{i_{\rm max}}(P)$ of path sets with size at 
	most $i_{\rm max}$, i.e., the common cumulants that we can ``observe'' using 
	empirical estimates. All other ``unobserved'' common cumulants are consigned to the 
	$\mathbf f_u$ vector. Note that $\mathbf f_o$ is not directly populated with common 
	cumulant estimates: in fact, both $\mathbf f_o, \mathbf f_u$ are left as decision 
	variables in the lasso optimization problem, and the value of $\mathbf f_o$ is 
	allowed to deviate from the empirical estimate if it promotes a sparser solution 
	$\mathbf g$. Instead, all of the empirical common cumulant estimates are collected 
	into a vector $\hat{\mathbf f_o}$, and the corresponding standard deviations of each 
	estimate are collected into the vector $\mathbf \sigma$. We then solve for the 
	optimal common cumulant vector $\mathbf f^* = \begin{pmatrix} 
	\mathbf f_o^* & \mathbf f_u^* \end{pmatrix}^\top$ using the convex, unconstrained 
	optimization problem:
	\begin{align} \label{eq:opt}
		\begin{split}
		\mathbf f_o^*, \mathbf f_u^* &= \argmin_{\mathbf f_o, \mathbf f_u} 
			J(\mathbf f_o, \mathbf f_u) \\
		J(\mathbf f_o, \mathbf f_u) &= || \mathbf \Sigma^{-1} (\mathbf f_o - \hat{\mathbf 
		f_o}) ||_2^2 + ||\mathbf D(\mathbf X_o \mathbf f_o + \mathbf X_u \mathbf 
		f_u)||_1
		\end{split}
	\end{align}
	Here $\mathbf \Sigma = \diag\{\mathbf \sigma\}$, and $\mathbf D$ is some tunable 
	diagonal 
	matrix of positive 
	weights (which we will soon discuss in more detail). Having computed the solution, we 
	then evaluate $\mathbf g^* = \mathbf X_o \mathbf f_o^* + \mathbf X_u \mathbf f_u^*$.
	
	Eqn. \eqref{eq:opt} simultaneously de-noises measurements of the observed common 
	cumulant values and imputes the unobserved common cumulants. The quadratic term 
	is proportional to the log likelihood of the data $\hat{\mathbf f_o}$ (under the 
	assumption of independent and normally-distributed common cumulant estimates with 
	variances $\mathbf \sigma^2$), and the regularizer $||\mathbf X_o \mathbf f_o + 
	\mathbf X_u \mathbf f_u||_1$ encourages sparsity in the vector $\mathbf g^*$. The end 
	result is an estimate of $\mathbf g_{i_{\rm max}}$ that only measures common cumulants
	up to a user-specified order and is more robust to noise in these measurements. 
	
	As with the full M\"obius Inference Algorithm, the columns of the routing matrix 
	correspond to the nonzero entries of $\mathbf g_{i_{\rm max}}$. Thus, once we obtain 
	an optimal (and sparse) exact cumulant vector $\mathbf g^*$, we add the 
	characteristic vector of 
	each $P \in \supp(\mathbf g^*)$ to our estimate of the routing matrix. 
	
	\paragraph*{Weighting the 1-Norm}
	A straightforward choice for weighting the 1-norm of $\mathbf g^*$ is to choose a 
	uniform weighting strategy, in which case $\mathbf D = \lambda \mathbf I$ for some 
	parameter $\lambda > 0$ that weights the 1-norm relative to the log 
	likelihood of the data. But uniform weighting tends to suppress entries of $\mathbf 
	g^*$ corresponding to singleton path sets. If $P = \{p\}$ for some $p \in \monpaths$, 
	then \eqref{eq:mobius-mod} shows that $g_i(P)$ is the only entry of $\mathbf g_{i}$ 
	that depends on $f_i(P)$. Thus, if the uncertainty $\sigma$ in the measurement of 
	$\hat f_o(P)$ is sufficiently large, the optimizer is free to zero out $g^*(P)$ by 
	tuning the decision variable corresponding to $f_i(P)$. Indeed, we have observed 
	numerically that uniform weighting leads to routing matrix estimates missing many 
	columns with single nonzero entries. 
	
	To counteract this problem, we suggest applying less weight to
	``under-determined'' entries of $\mathbf g^*$. Formally, for each $P$ in the 
	support estimate, let 
	\[
		a(P) = \begin{cases}
			\left| \left\{ Q ~\text{in supp. est.} : \mathbf X_o(Q, P) > 0 
			\right\} \right|, & |P| \le i_{\rm max} \\
			\left| \left\{ Q ~\text{in supp. est.} : \mathbf X_u(Q, P) > 0 
			\right\} \right|, & |P| > i_{\rm max} \\
		\end{cases}
	\]
	be the number of entries of $\mathbf g^*$ that depend on the decision variable 
	corresponding to $f_{i_{\rm max}}(P)$. We then choose the weight corresponding to 
	$g^*(P)$ according to $d(P) = \lambda a(P)^b$, where $\lambda > 0$ is a uniform 
	overall weight for the 1-norm term, and $b \in [0, 1)$ is some exponent. The exponent 
	should be non-negative to ensure that the weight is increasing in $a(P)$, but it 
	should also be fairly small, so that the weight's rate of change rapidly tapers off 
	for positive $a(P)$. We have found empirically that setting $b$ between 0.2 and 0.4 
	is generally a good choice. 
	
	\subsection{Putting Everything Together}
	
	\tikzstyle{block} = [rectangle, draw,
	text width=5em, text centered, rounded corners, minimum height=3em]
	\tikzstyle{line} = [draw, -latex']
	\tikzstyle{cloud} = [draw, ellipse, node distance=2cm, minimum height=2em]
	
	\begin{figure}
		\centering
		\begin{tikzpicture}[node distance = 2cm, auto]
		\node [block] (s1) {Stage 1};
		\node [block, right of=s1, node distance=3cm] (s2) {Stage 2};
		\node [block, right of=s2, node distance=3cm] (s3) {Stage 3};
		\node [block, above of=s2, node distance=2.5cm] (e) {$k$-statistics};
		\node [block, below of=s3, node distance=2cm] (f) {$\hat{\mathbf R}$};
		\node [cloud, above of=e, node distance=1.5cm] (d) {Path Delay Data};
		\node [cloud, below of=s2, node distance=2cm] (p) {User Parameters};
		\path [line] (s1) -- node{$\mathcal B$} (s2);
		\path [line] (s2) -- node{$\mathbf X$} (s3); 
		\path [line] (e) -- node[left]{$\hat f_i(P), \text{Var}(\hat 
			f_i(P))$} (s3); 
		\path [line] (e) -- node[left]{$\texttt{Nonzero}(f_i(P))$} (s1);
		\path [line] (d) -- (e);
		\path [line] (s3) -- node{$\mathbf g^*$} (f); 
		\path [line] (p) -- node[left]{$\mathcal B_0, i_0, i_f, t$\;\;} (s1);
		\path [line] (p) -- node[left]{$s$} (s2);
		\path [line] (p) -- node[right]{\;\;$i_{\rm max}$} (s3);
		\end{tikzpicture}
		\caption{\rev{Diagram of the Sparse M\"obius Inference procedure.}}
		\label{fig:diagram}
	\end{figure}
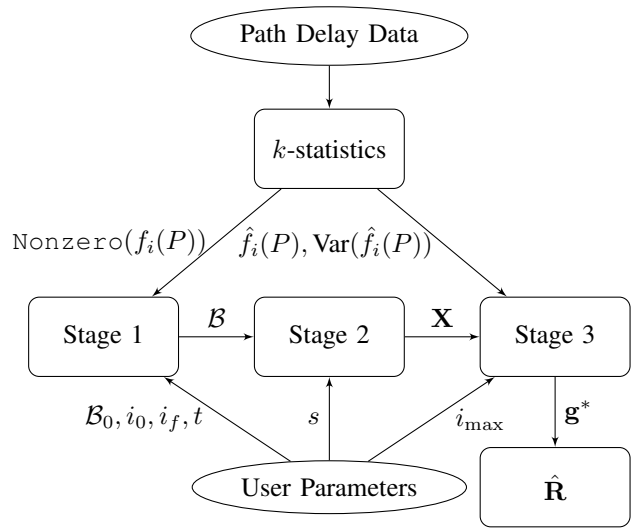
	
	For completeness, we now show how the three stages of the Sparse M\"obius Inference 
	procedure come together to form a data-to-routing-matrix pipeline. Figure 
	\ref{fig:diagram} depicts a diagram of this process.
	
	The user begins Stage 1 with an initial guess of the bounding topology $\mathcal B_0 
	\subseteq 2^{\monpaths}$ (either $\{\monpaths\}$ or maximal cliques of the graph 
	formed by nonzero covariances), an initial cumulant order $i_0$ (usually 2 or 3), a 
	final cumulant order $i_f$ (e.g., 4 or 5), and a threshold function $t$ (perhaps 
	using quantiles of the binomial distribution). Algorithm \ref{alg:bound} then 
	tightens the support estimate by setting $\mathcal B = 
	\texttt{BoundingTopology}(\mathcal B_0, i_0, i_f, t)$, using the path delay dataset 
	to evaluate $\texttt{Nonzero}(f_i(P))$ for orders $i = i_0, i_0 + 1, \dots, i_f$. 
	Then $\mathcal B$ is passed on to Stage 2.
	
	In the second stage, the user provides a size threshold $s$ for the ``hard'' sparsity 
	heuristic. In accordance with \eqref{eq:mobius-mod}, the modified M\"obius inversion 
	matrix $\mathbf X$ is constructed, considering only rows and columns of the matrix 
	corresponding to path sets in the support estimate of $\mathcal B$ that are either 
	directly in $\mathcal B$ or at most of size $s$. This matrix $\mathbf X$ is passed to 
	Stage 3.
	
	To begin the final stage, the user specifies a cumulant order $i_{\rm max}$ (e.g., 3, 
	4, or 5) and partitions the common cumulant vector and the matrix $\mathbf X$ 
	accordingly. For path sets of size at most $i_{\rm max}$, the path delay data is 
	once again used to estimate the common cumulants $\hat{\mathbf f_o}$ and the 
	variances $\mathbf \sigma^2$ of these estimates. Solving \eqref{eq:opt} yields a 
	filtered common cumulant vector $\mathbf f^*$, leading to a sparse estimate $\mathbf 
	g^* = \mathbf X \mathbf f^*$ of the exact cumulant vector. Finally, the routing 
	matrix estimate $\hat{\mathbf R}$ is constructed from the zero-nonzero pattern of 
	$\mathbf g^*$.  
 
 	\begin{figure*}
 		\centering
 		\includegraphics[width=0.9\linewidth]{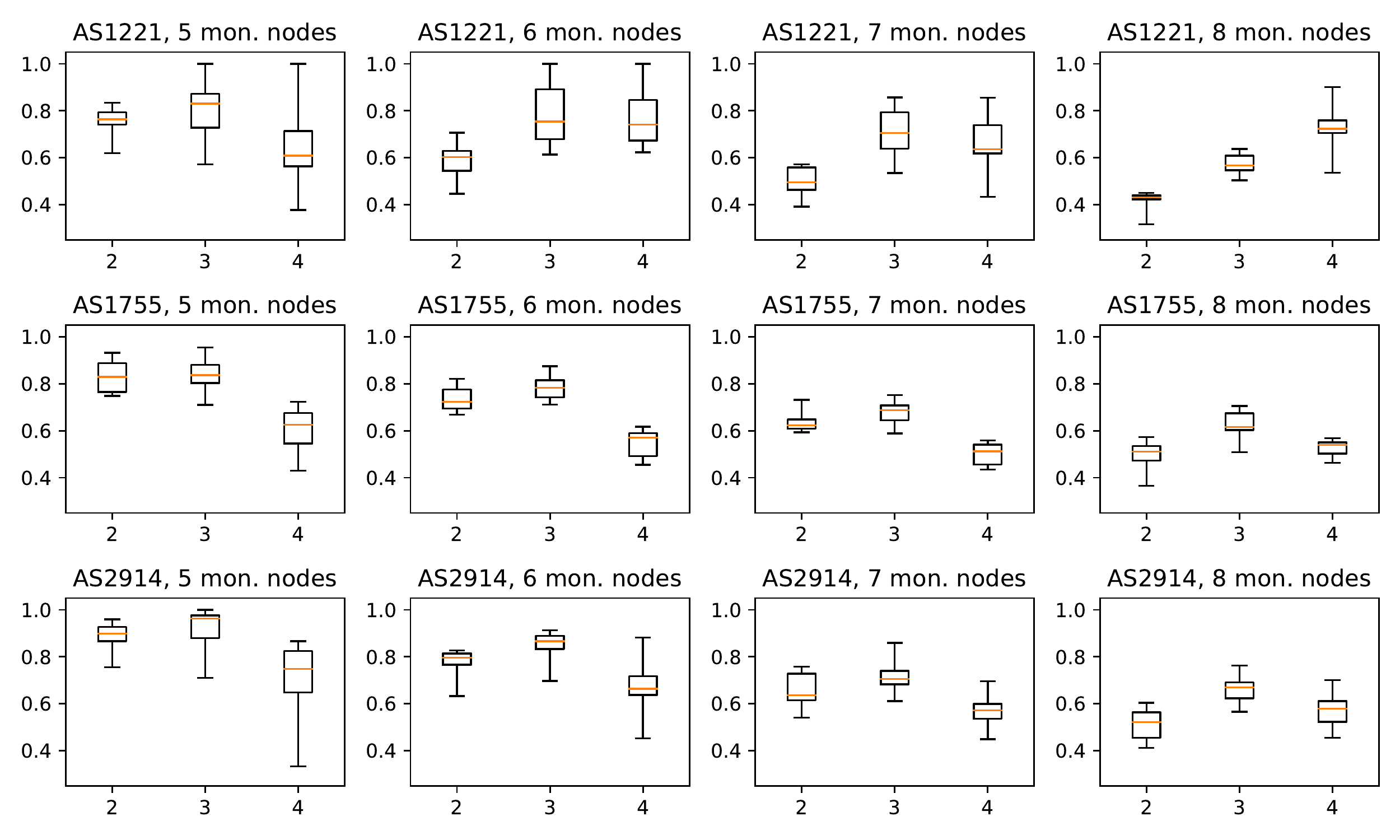}
 		\caption{\rev{Distributions of F1 scores of the routing matrix estimate for the 
 		120 
 		case studies, based on a sample of size 100,000. Plots in each row are based on 
 		the same underlying network, and plots in the same column have the same number of 
 		monitor nodes. The three boxes in each plot correspond to values $i_{\rm max} = 
 		2, 3, 4$ used for inference.}}
 		\label{fig:f1-scores}
 	\end{figure*}

 	\section{Results and Evaluation}
 	\label{sect:results}
 	
 	What follows is an abbreviated set of experimental results applying Sparse M\"obius 
 	Inference to many synthetic datasets. The full description of our methodology and 
 	results are contained in Appendix A (in the supplementary file).
 	
 	\paragraph*{Synthetic Datasets}
 	We created 120 synthetic datasets based on real ISP network topologies, provided by 
 	Rocketfuel \cite{NS-RM-DW:02}. We selected three networks within the Rocketfuel 
 	database with 
 	different sizes and densities (AS1221, AS1755, and AS2914). For each topology, we 
 	generated 40 synthetic datasets 
 	of path delays: 10 each for experiments with 5, 6, 7, and 8 monitor nodes. For 
 	each of these 40 case studies, the network links are assigned different gamma delay 
 	distributions, the $n_{\rm node}$ monitor nodes are selected at random, and the 
 	$n = \binom{n_{\rm node}}{2}$ monitor paths are chosen by computing the shortest path 
 	between each pair of monitor nodes. Then a large sample of the joint path delay 
 	distribution is recorded.
 	
 	\paragraph*{Sparsity of the Common and Exact Cumulants}
 	The Sparse Mo\"bius Inference procedure is based on the postulate that the vectors of 
 	common and exact cumulants are both sparse. This assumption holds up extremely well 
 	in our case studies; with $n=28$ paths, for example, 99.99\% to 99.999\% of 
 	the entries of the common cumulant vector are zero. 
 	
 	\paragraph*{Evaluating the Bounding Topology}
 	The first stage of Sparse M\"obius Inference uses low-order cumulants to estimate 
 	$\supp(\mathbf f_i)$. Our results indicate that Algorithm \ref{alg:bound} is very 
 	effective at finding a bounding topology with a tight support estimate. For almost 
 	all of the 120 case studies, third-order cumulants $(i_f = 3)$ with a sample size $N 
 	= 50,000$ or larger are sufficient to construct a bounding topology that predicts 
 	$\supp(\mathbf f_i)$ with an F1 score of 1.0 (or extremely close to 1.0). 

 	\paragraph*{Evaluating the Estimated Routing Matrix}
 	Next, we evaluate the performance of Sparse M\"obius Inference end-to-end. We ran 
 	stages 2 and 3 to get an estimate of $\hat{\mathbf R}$ for each case study and 
 	various sample sizes, using as input to Stage 2 the bounding topologies computed with 
 	$i_f = 4$ from the same sample. The hyperparameters of the lasso heuristic ($\lambda$ 
 	and the exponent $b$) are tuned separately for each underlying network and number of 
 	monitor paths. Figure \ref{fig:f1-scores} shows the F1 scores that we obtained for 
 	each of the 120 case studies. For all underlying networks, the performance tends to 
 	degrade with the number of monitor paths, and the best estimate is usually obtained 
 	using third-order $k$-statistics ($i_{\rm max} = 3$). 
 	
 	\paragraph*{Evaluating the Lasso Heuristic}
 	We also evaluated the lasso heuristic in Stage 3 using ground-truth cumulants. For 
 	these experiments, we borrowed the bounding topologies computed from the $N = 
 	100,000$ sample with $i_f = 4$, but instead of populating the $\hat{\mathbf 
 	f}_o$ vector in \eqref{eq:opt} with $k$-statistics computed from this sample, we used 
 	the true common cumulants. These values have no uncertainty, so we removed the 
 	quadratic penalty from $J(\mathbf f_o, \mathbf f_u)$, instead constraining $\mathbf 
 	f_o = \hat{\mathbf f}_o$. Again, the hyperparameters $\lambda$ and $b$ are tuned 
 	separately for each network and number of monitor paths.
 	Figure \ref{fig:f1-scores-inf} plots the distribution of the resulting F1 scores. For 
 	smaller (5 or 6 monitor) scenarios, the lasso heuristic is typically capable of 100\% 
 	accurate routing matrix reconstruction. For larger scenarios, the heuristic requires 
 	up to third-order cumulants for completely accurate inference. 
 	
 	\begin{figure}
 		\includegraphics[width=\linewidth]{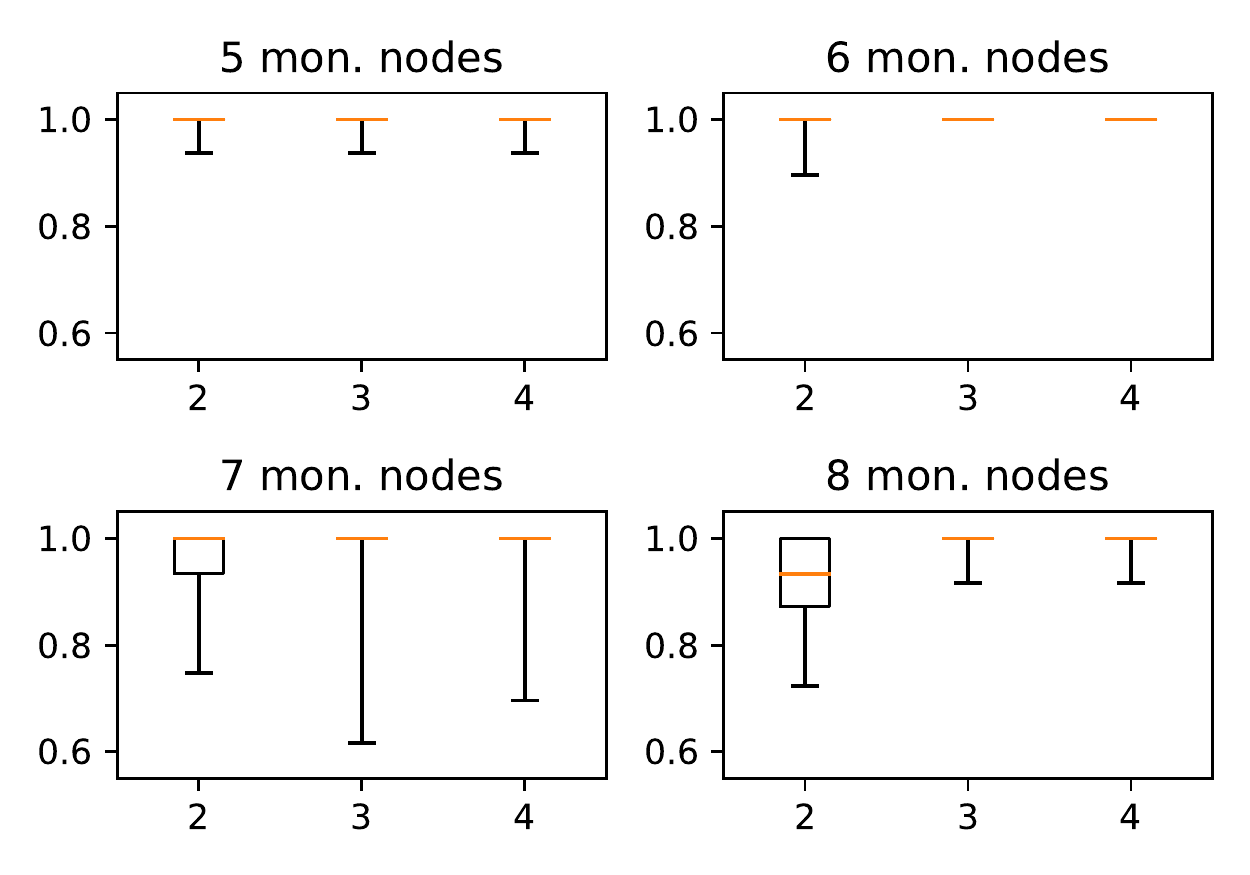}
 		\caption{\rev{Distributions of F1 scores of the routing matrix estimate based on 
 		ground-truth cumulants (instead of $k$-statistics). Each plot corresponds to a 
 		particular number of monitor paths, and the results are aggregated across case 
 		studies from the 3 underlying networks. The three boxes in each plot correspond 
 		to values $i_{\rm max} = 2, 3, 4$ used for inference.}}
 		\label{fig:f1-scores-inf}
 	\end{figure}
 	
 	\paragraph*{Discussion}
 	Our results paint a mixed but optimistic picture for the Sparse M\"obius Inference 
 	procedure. Admittedly, higher F1 scores from the $N = 100,000$ sample would be 
 	desirable before the method is deployed in real-world applications. But the two key 
 	components of the procedure---estimating $\supp(\mathbf f_i)$ from 
 	low-order $k$-statistics, and using the lasso sparsity heuristic to infer $\mathbf R$ 
 	without using the high-order cumulants required by MIA---worked very well in 
 	isolation, achieving 100\% accuracy in most scenarios. 
 	
}
\rev{
 	
 	\section{Conclusion}
 	We have provided a novel tomographic approach to routing topology inference from path 
 	delay data, without making any assumptions on routing behavior. Through MIA, we have 
 	provided a theoretical framework for extending the use of second-order statistics in 
 	network tomography toward higher-order statistics. Furthermore, we have introduced 
 	the Sparse M\"obius Inference procedure, which implements a heuristic and more 
 	practical variant of MIA. We have extensively studied the performance of Sparse 
 	M\"obius Inference using many synthetic case studies. While more work is needed to 
 	improve the filtering of noisy $k$-statistics, our results indicate that 
 	the Sparse M\"obius Inference can serve as a solid foundation for future 
 	improvements. }

	\bibliographystyle{IEEEtran}
	\bibliography{alias,Main,FB} 
	

	\end{document}


\title{Appendices to ``Topology Inference with Multivariate Cumulants:
	The M\"obius Inference Algorithm''}
	\author{Kevin D. Smith, Saber Jafarpour, Ananthram Swami, and Francesco Bullo}
	
	\maketitle
	
	\appendix
	\section{Extended Results and Evaluation}
	
	\subsection{Synthetic Datasets}
	
	To evaluate the Sparse M\"obius Inference procedure, we created several synthetic 
	datasets based on real ISP network topologies, provided by Rocketfuel 
	\cite{NS-RM-DW:02}. We selected three networks within the Rocketfuel database with 
	different sizes and densities: AS1221 (the Telstra network in Australia), AS1755 (the 
	EBONE network in Europe), and AS2914 (the Verio network in the United States). Table 
	\ref{table:nets} contains some summary statistics of these networks. AS1755 is the 
	least ``dense,'' in that it has the smallest average degree and the smallest 
	clustering coefficient. AS2914, on the other hand, has by far the largest average 
	degree (and a similar clustering coefficient to AS1221).
	
	\begin{table}
		\centering
		\begin{tabular}{lrrrr}
			Network & Nodes & Links & Avg. Degree & Cluster Coef \\
			\hline
			AS1221 & 318 & 758 & 4.77 & 0.28 \\
			AS1755 & 172 & 381 & 4.43 & 0.20 \\
			AS2914 & 960 & 2821 & 5.88 & 0.25
		\end{tabular}
		\caption{Topological properties of the networks, including number of nodes and 
		links (edges), the average node degree, and the graph clustering 
		coefficient.}
		\label{table:nets}
	\end{table}
	
	\begin{figure}
		\centering
		\includegraphics[clip,trim={0.5in 0 0 0},width=0.5\linewidth]{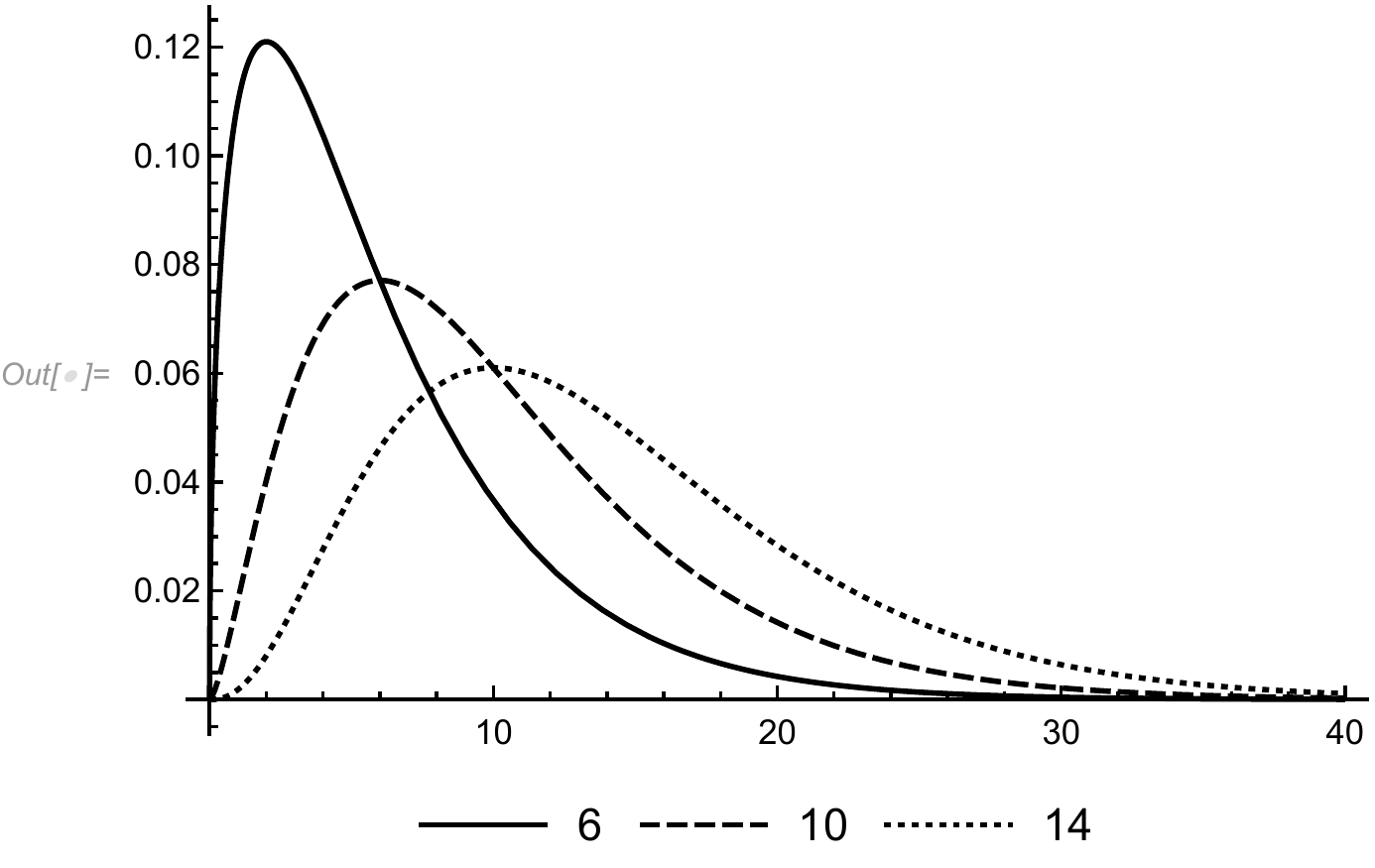}
		\caption{Sample link distributions with delay means of 6 ms, 10 ms, and 14 ms.}
		\label{fig:link-dist}
	\end{figure}
	
	For each of these 3 underlying networks, we generated 40 synthetic datasets 
	of path delays: 10 each for experiments with 5, 6, 7, and 8 monitor nodes, resulting 
	in a total of 120 case studies. In each case study, the network links are assigned 
	different delay distributions. The mean delay $\mu$ on each link is chosen from a 
	normal distribution with a mean of 10 ms and a standard deviation of 2 ms. Based on 
	$\mu$, the link delay is assigned a gamma distribution with shape parameter $\alpha = 
	\frac \mu 4$ and rate 
	parameter $\beta = \frac 1 4$. Figure \ref{fig:link-dist} plots some samples of these 
	delay distributions. Then for each of these 120 cases studies, the $n$ monitor nodes 
	are selected at random, and the $\binom{n}{2}$ monitor paths are chosen by computing 
	the shortest path between each pair of monitor nodes (based on the mean link delays).
	
	Finally, for each of the 120 case studies, we generate three different samples of the 
	joint path delay distribution with sizes 10,000, 50,000, and 100,000. For each of 
	these three values of $N$, the sample is bootstrapped into 50 different 
	``re-samples'' of size $N$ (by randomly sampling the original $N$ points with 
	replacement). Note that bootstrapping does not introduce new data; rather, the 50 
	different re-samples of the original $N$ points allow us to empirically estimate the 
	distribution of $k$-statistics. (We tried using 100 re-samples instead of 50 as well, 
	but only led to marginal improvement with significantly greater runtime.)
	
	\subsection{Sparsity of the Common and Exact Cumulants}
	
	\begin{figure}
		\centering
		\includegraphics{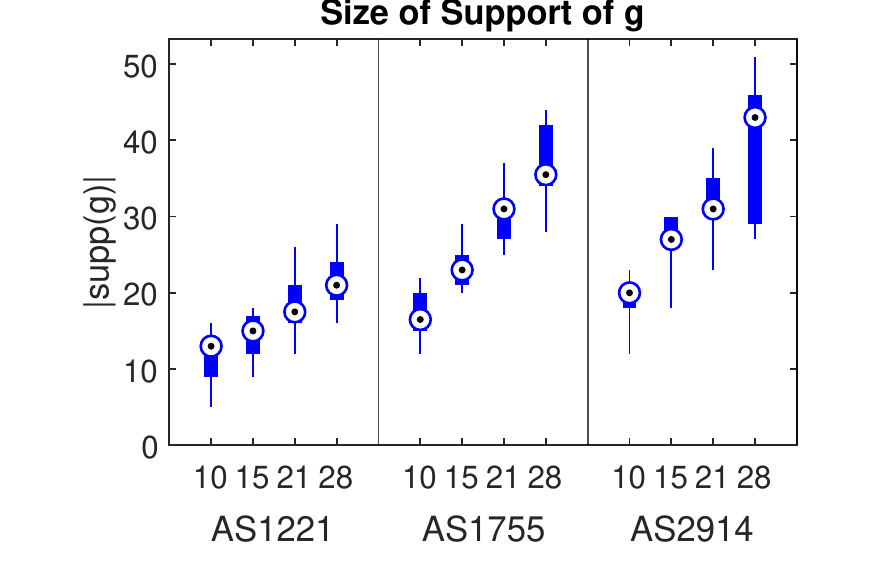}%
		\includegraphics{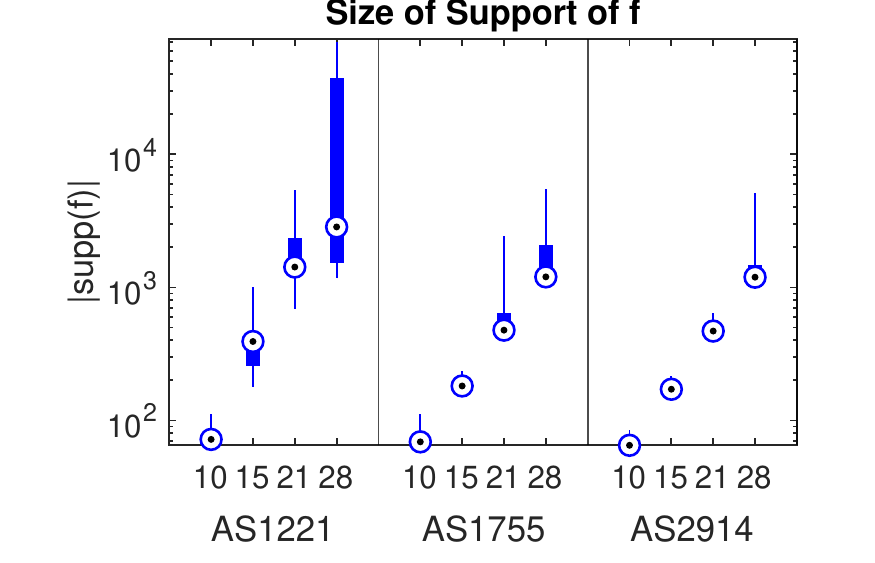}
		
		\includegraphics{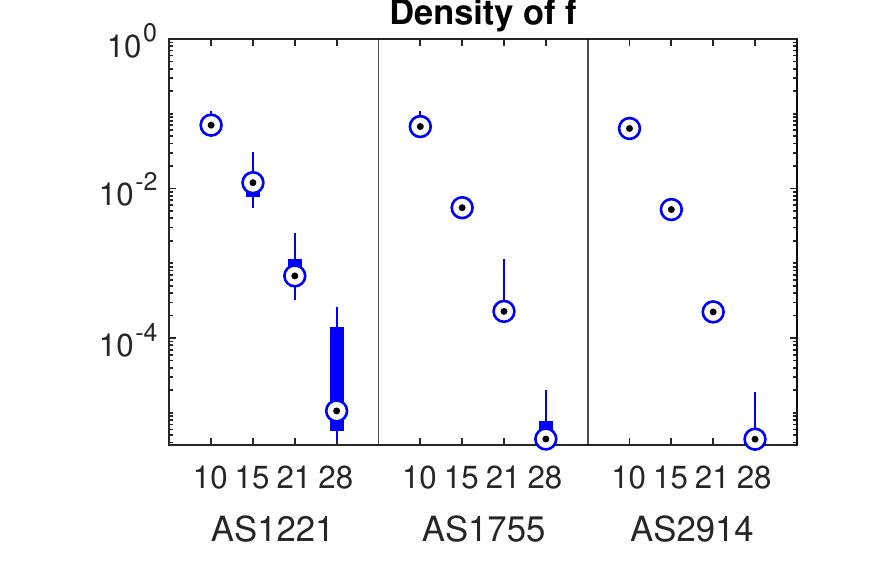}%
		\includegraphics{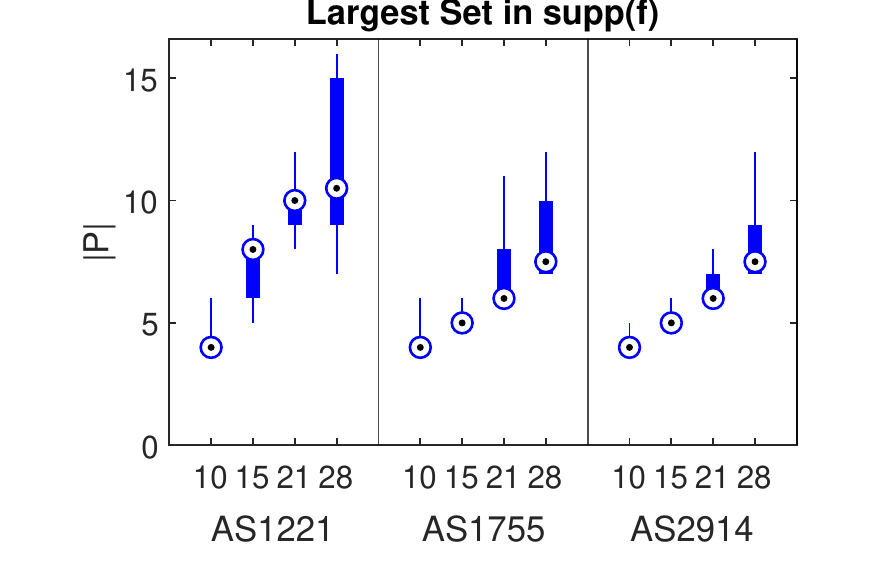}
		
		\caption{Sparsity metrics for the common and exact cumulant vectors. Each box 
			reflects the distribution across 10 case studies for each topology and number 
			of 
			monitor paths. Note that the size of the support of $\mathbf f_i$ (upper 
			right) 
			and the density of $\mathbf f_i$ (lower left) are plotted with a log scale on 
			the 
			vertical axis. The horizontal axis is \textit{not} drawn with proportional 
			spacing between the numbers of monitor paths.}
		\label{fig:sparsity}
	\end{figure}
	
	The Sparse Mo\"bius Inference procedure is based on the postulate that the vectors of 
	common and exact cumulants are both sparse, so we ought to examine how well this 
	premise holds up in our empirical case studies. Figure \ref{fig:sparsity} plots the 
	following four sparsity metrics for the 120 case studies: 
	\begin{itemize}
		\item Size of $\supp(\mathbf g_i)$, i.e., the number of path sets $P \subseteq 
		\monpaths$ for which $g_i(P) \ne 0$. This is equivalent to the number of logical 
		links.
		\item Size of $\supp(\mathbf f_i)$, i.e., the number of path sets $P \subseteq 
		\monpaths$ for which $f_i(P) \ne 0$.
		\item Density of $\supp(\mathbf f_i)$, i.e., the fraction $(2^n - 1)^{-1} 
		|\supp(\mathbf f_i)|$. 
		\item Largest set in $\supp(\mathbf f_i)$, i.e., the size $|P|$ of the largest 
		set $P \subseteq \monpaths$ for which $f_i(P) \ne 0$.
	\end{itemize} 
	Unsurprisingly, the number of logical links $|\supp(\mathbf g_i)|$ tends to be small 
	compared to the number of path sets $2^n - 1$, and the number of links increases in a 
	roughly linear manner with the number of monitor paths. Figure \ref{fig:sparsity} 
	(upper left) also shows that the number of links utilized by $n$ monitor paths 
	depends heavily on the underlying network topology. AS2914 leads to the largest 
	numbers of logical links, which is to be expected, since AS2914 is the largest of the 
	three networks. Remarkably, AS1221 has by far the smallest numbers of links, even 
	though AS1221 has the middle number of links and average degree and the largest 
	clustering coefficient. 
	
	Figure \ref{fig:sparsity} (upper right and lower left) also indicate that the common 
	cumulant vector is very sparse. While the number of nonzero entries increases roughly 
	exponentially in $n$, the fraction of the $2^n - 1$ entries which are nonzero also 
	decreases rapidly in $n$, approaching 0.001\% density for scenarios with 28 monitor 
	paths. Evidently the first stage of the Sparse M\"obius Inference procedure is very 
	well justified in trying to isolate the support of $\mathbf f_i$, as this stage (if 
	it is accurate) will eliminate the vast majority of the common cumulant entries from 
	the problem, thereby greatly reducing the number of $k$-statistics that need to be 
	evaluated. We should point out that this measurement task can still be nontrivial: as 
	the upper right plot indicates, the number of entries in the common cumulant vector 
	was close to 100,000 in the ``worst'' of our case studies. This provides some 
	additional motivation for the final stage of the Sparse M\"obius Inference procedure, 
	which leaves many entries of the common cumulant un-measured and infers them instead 
	through the lasso heuristic. 
	
	Finally, Figure \ref{fig:sparsity} (lower right) depicts how large the path sets in 
	$\supp(\mathbf f_i)$ can get. One way to interpret this sparsity metric is how 
	``crowded'' by monitor paths the links can get. AS1221 tends to have the largest path 
	sets with nonzero common cumulants. One scenario with $n=28$ even has a path set of 
	size 16 in $\supp(\mathbf f_i)$, which implies that 16 of the 28 monitor paths in 
	this scenario all use one particular link. This ``crowding'' of links in AS1221 
	naturally complements the small sizes of $|\supp(\mathbf g_i)|$ in the upper left 
	plot: since the monitor paths traverse a smaller number of links in this network, it 
	makes sense that the few links that are utilized will have to withstand heaver 
	utilization by the monitor paths. 
	
	\subsection{Evaluating the Bounding Topology}
	
	The first stage of Sparse M\"obius Inference (described in Section V.A) uses 
	low-order $k$-statistics to estimate $\supp(\mathbf f_i)$, by constructing a bounding 
	topology $\mathcal B$. For each of the 120 case studies and each of the 3 sample 
	sizes, we ran Algorithm 3 to construct a bounding topology. We implemented the 
	$\texttt{Nonzero}(f_i(P))$ hypothesis test with the following procedure: (i) compute 
	50 estimates of $\hat f_i(P)$ by applying (12) to each of the bootstrapped 
	re-samples; (ii) estimate a $p$-value for the null hypothesis that $f_i(P) = 0$ by 
	applying a one-sample Student's $t$-test to the 50 estimates of $\hat f_i(P)$; and 
	(iii) deciding $\texttt{Nonzero}(f_i(P))$ is true if and only if that $p$-value is 
	below a pre-determined threshold, $\alpha$. The values used for these thresholds are 
	reported in Table \ref{table:bounding-topo-params}.
	
	\begin{table}
		\centering
		\begin{tabular}{r|r|rrr|rrr}
			& \multicolumn{1}{c}{$i = 2$} & \multicolumn{3}{c}{$i = 3$} & 
			\multicolumn{3}{c}{$i = 4$} \\
			$N$ & $\alpha$ & $\alpha$ & $\beta$ & $\gamma$ & $\alpha$ & $\beta$ & 
			$\gamma$ \\ 
			\hline
			10,000 & $10^{-20}$ & $10^{-10}$ & 0.1 & 0.15 & $10^{-2}$ & 0.25 & 0.3 \\ 
			50,000 & $10^{-40}$ & $10^{-30}$ & 0.05 & 0.15 & $10^{-5}$ & 0.05 & 0.15 \\
			100,000 & $10^{-40}$ & $10^{-30}$ & 0.05 & 0.15 & $10^{-10}$ & 0.05 & 0.15 \\
		\end{tabular}
		\caption{Parameters used for estimating the bounding topology: $\alpha$ (the 
		$p$-value threshold for the $\texttt{Nonzero}$ hypothesis test), $\beta$ (an 
		estimate for the Type II error of the $\texttt{Nonzero}$ test), and $\gamma$ (a 
		desired upper bound for the probability of falsely removing a size-$i$ subset of 
		$\monpaths$ from the support estimate). Different parameters are used for 
		different sample sizes $N$, corresponding to each row. Different parameters are 
		also used for different $k$-statistic orders $i$, corresponding to each group of 
		columns.} 
		\label{table:bounding-topo-params}
	\end{table}
	
	The initial bounding topologies $\mathcal B_0$ were constructed using second-order 
	$k$-statistics (covariances) and the maximal clique approach described in Section 
	V.A. We recorded these initial values of $\mathcal B_0$ to assess the accuracy of 
	their support estimates, and then we applied Algorithm 2 to tighten the bounding 
	topology using $k$-statistics of order $i = 3$. We recorded these ``third-order'' 
	estimates as well, and then applied Algorithm 2 one last time with $k$-statistics of 
	order $i = 4$. Again, the final bounding topology is recorded and the accuracy of its 
	support estimate is assessed.
	
	The threshold function $t$ that we supplied to Algorithm 2 was constructed from the 
	quantile approach described in Section V.A. Recall the quantile approach has two 
	parameters: $\beta$, an estimate for the Type II error rate of 
	$\texttt{Nonzero}(f_i(P))$; and $\gamma$, an upper bound for the probability that a 
	size-$i$ path set is incorrectly removed from the support estimate. The optimal 
	values of these parameters depend on the $k$-statistic order and sample size, since 
	higher-order $k$-statistics tend to be less accurate, and larger samples lead to 
	lower variance of the $k$-statistic estimates. We tuned these parameters to 
	the values reported in Table \ref{table:bounding-topo-params}. 

	To evaluate the performance of this stage, we study how the support estimate after 
	each order $i$ (2, 3, and 4) compares to the ground-truth support of $\mathbf f_i$, 
	as well as how this performance scales with the size of the path delay sample. We 
	assess the accuracy of the support estimate using two standard metrics for binary 
	classifiers: \textit{precision}, i.e., the 
	fraction of path sets in the support estimate that truly belong to 
	$\supp(\mathbf f_i)$; and \textit{recall}, the fraction of $\supp(\mathbf f_i)$ that 
	is in the support estimate. For good performance of the first stage, the recall 
	should be very close to one (so that non-zero entries of $\mathbf f_i$ are not 
	ignored in the next two stages), and the precision should approach one as the support 
	estimate is tightened with successive orders $i$. 
	
	Figure \ref{fig:bounding-topo-prec-full} shows the precision of the support estimate 
	after each successive round of tightening, while Figure 
	\ref{fig:bounding-topo-rec-full} shows the corresponding recall. For the vast 
	majority of the 120 case studies, using $N = 50,000$ samples and a $k$-statistic 
	order up to $i = 3$ is enough to get a 100\% accurate support estimate. There is also 
	very little difference in the results between using the $N = 50,000$ and $N = 
	100,000$, indicating that a sample size of 50,000 is usually enough to determine 
	whether or not $f_i(P) = 0$ for orders up to $i = 4$. Overall, the results show that 
	Stage 1 of Sparse M\"obius Inference is highly successful in identifying the 
	collection of path sets with a nonzero common cumulant.
	
	\begin{figure}
		\centering
		\includegraphics{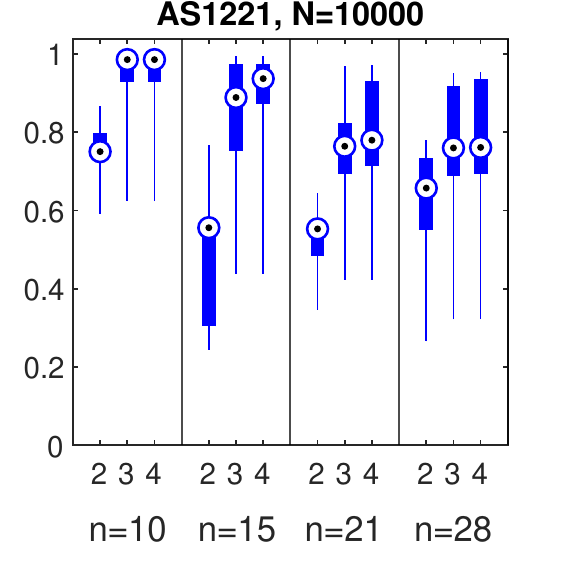}%
		\includegraphics{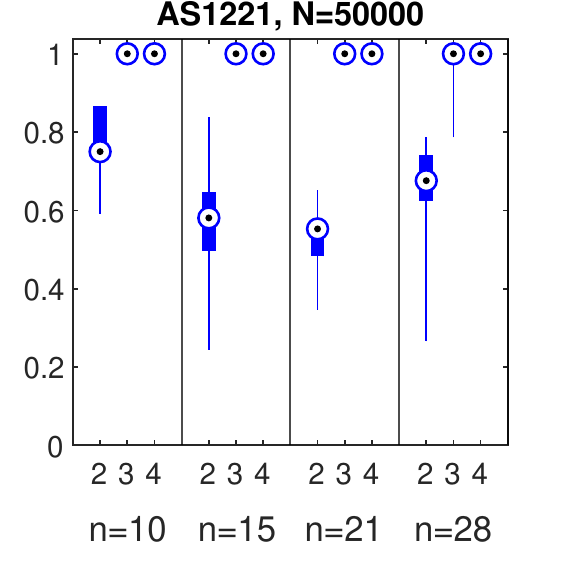}%
		\includegraphics{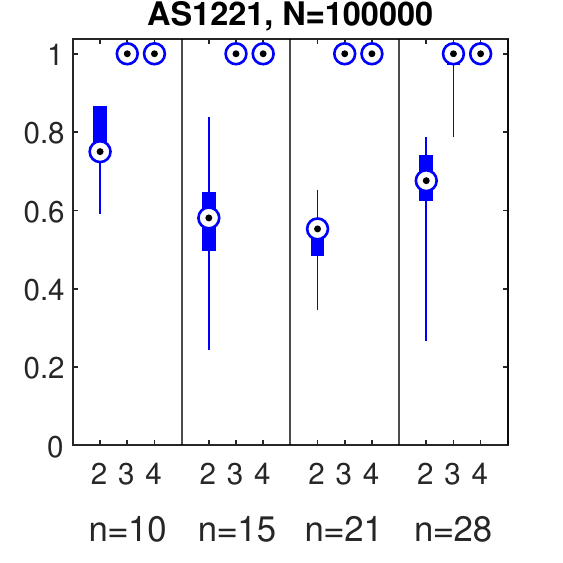}
		\includegraphics{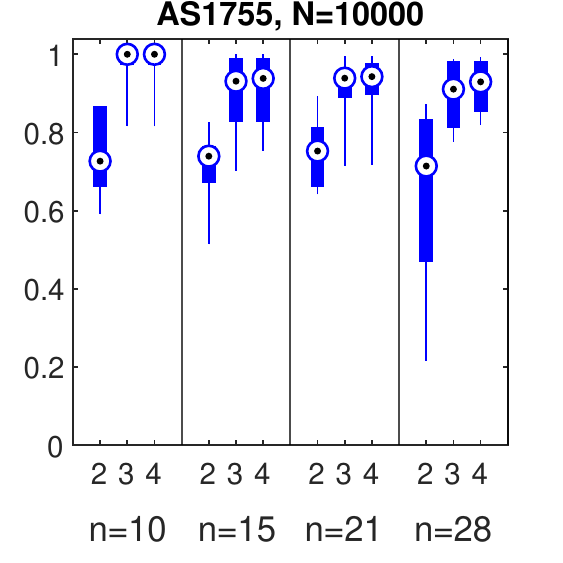}%
		\includegraphics{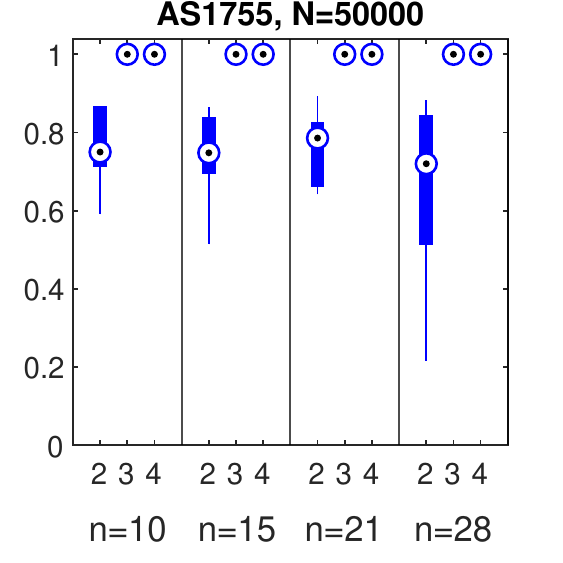}%
		\includegraphics{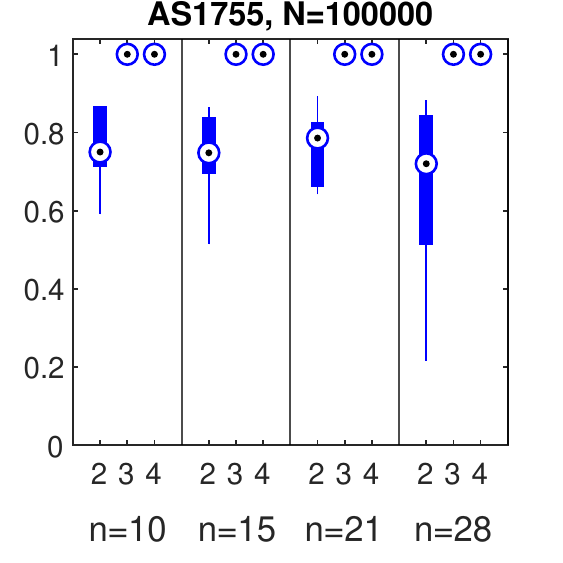}
		\includegraphics{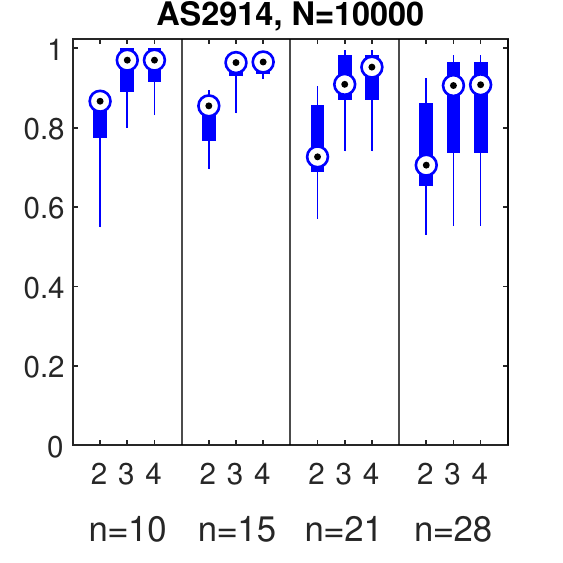}%
		\includegraphics{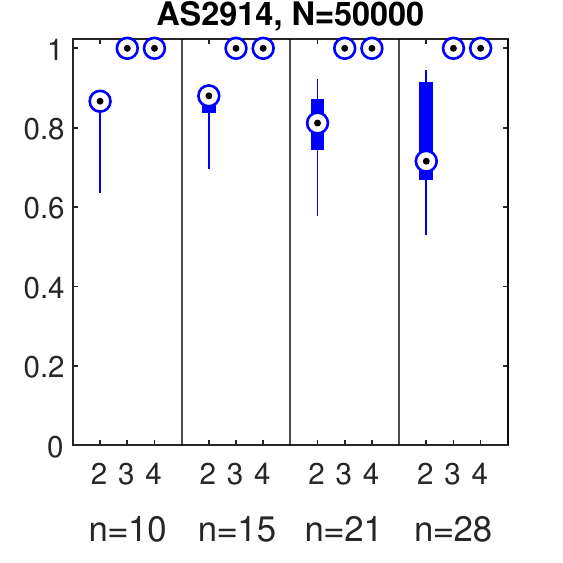}%
		\includegraphics{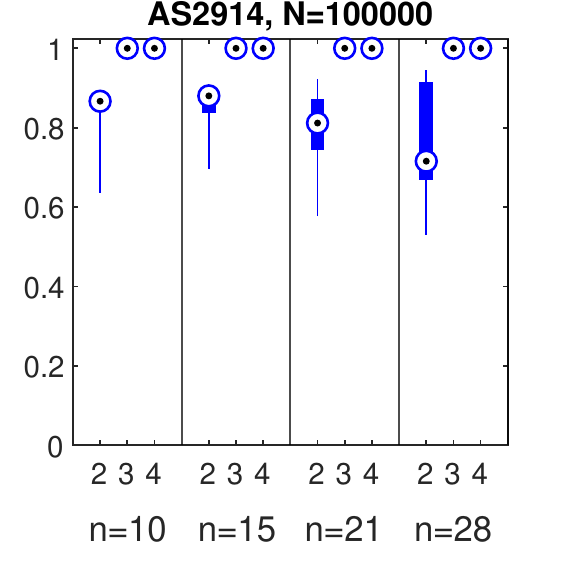}
		\caption{Precisions of the bounding topology estimates. Each plot corresponds to 
		a particular underlying network and sample size. Plots are sub-divided by the 
		number of monitoring paths $n$. For each network and for each $n$, samples are 
		drawn from 10 randomly-initiated tomography scenarios, and the bounding topology 
		is computed using cumulant orders 2, 3, and 4. The distributions of precisions 
		for these 10 bounding topologies are depicted using a box and whisker plot. 
		Circular markers represent the median.}
		\label{fig:bounding-topo-prec-full}
	\end{figure}

	\begin{figure}
		\centering
		\includegraphics{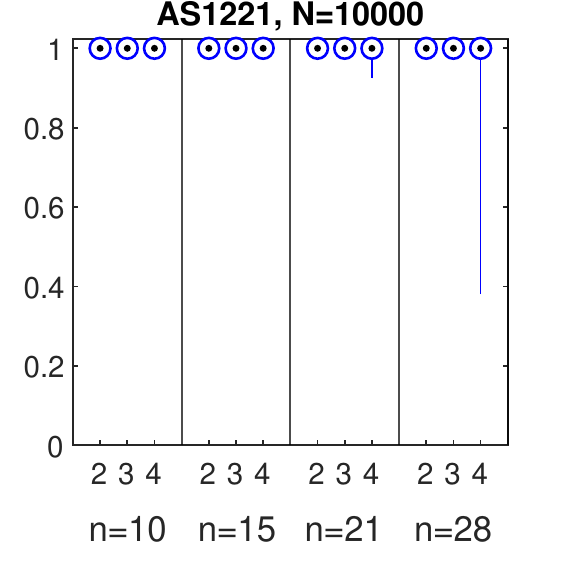}%
		\includegraphics{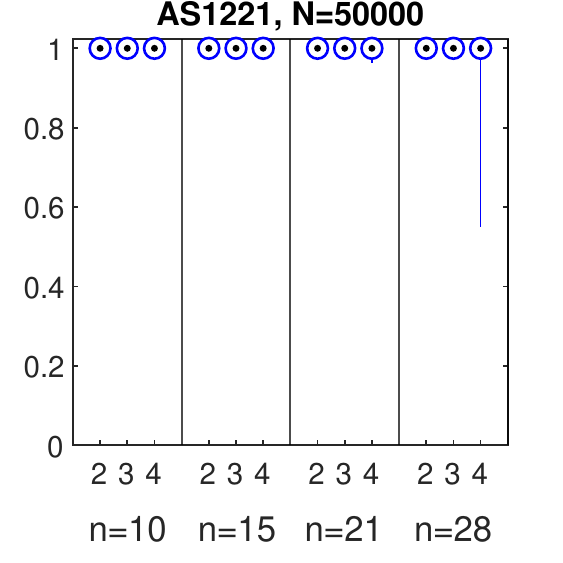}%
		\includegraphics{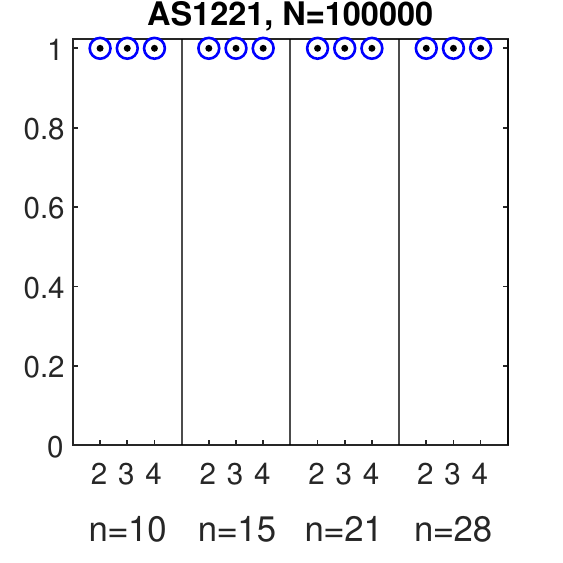}
		\includegraphics{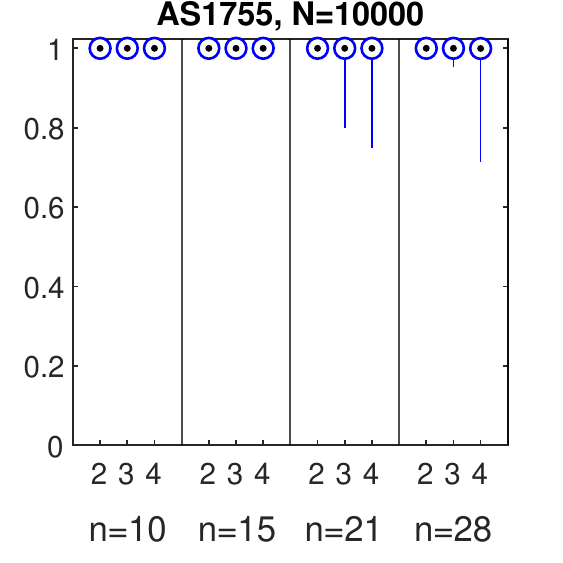}%
		\includegraphics{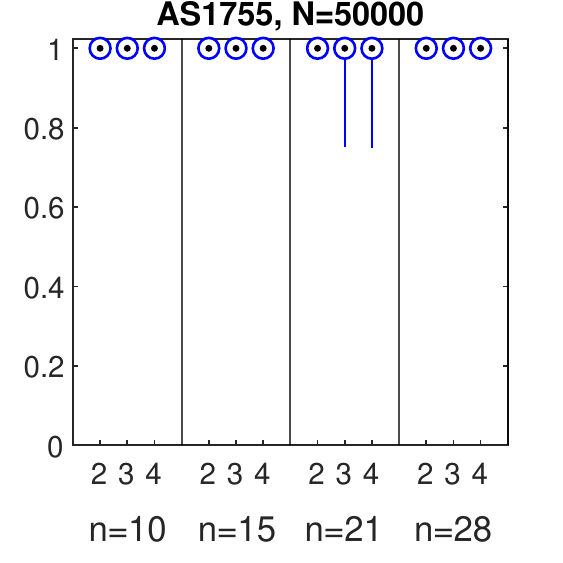}%
		\includegraphics{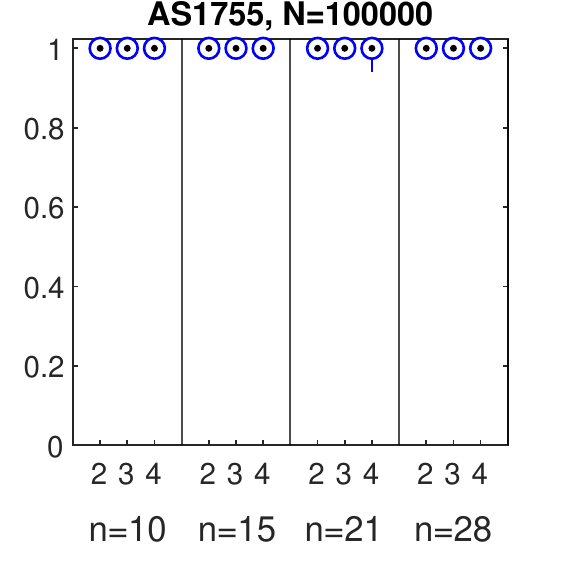}
		\includegraphics{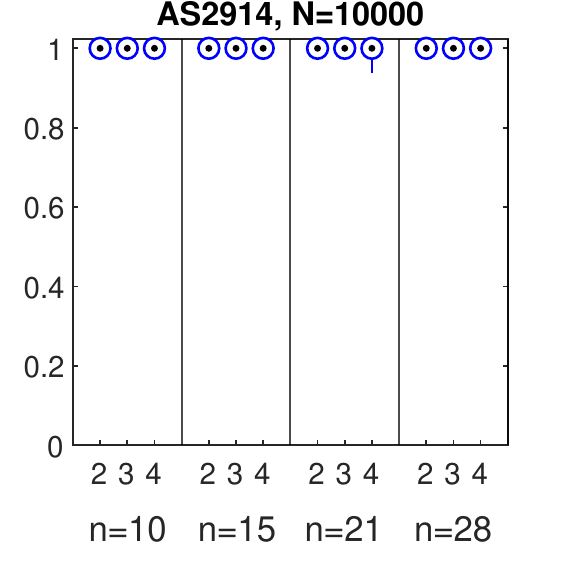}%
		\includegraphics{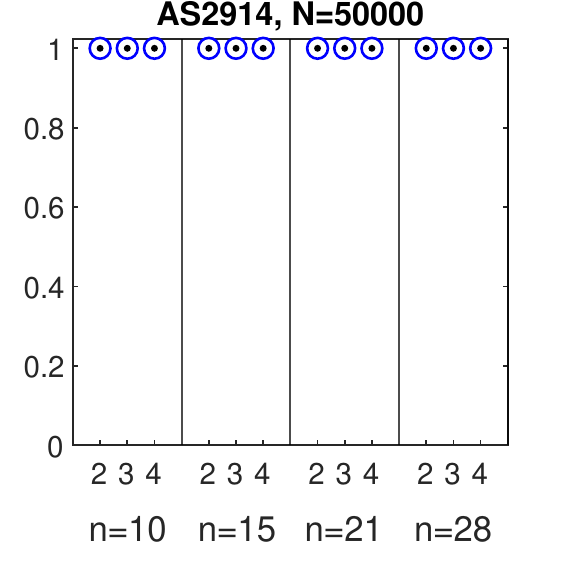}%
		\includegraphics{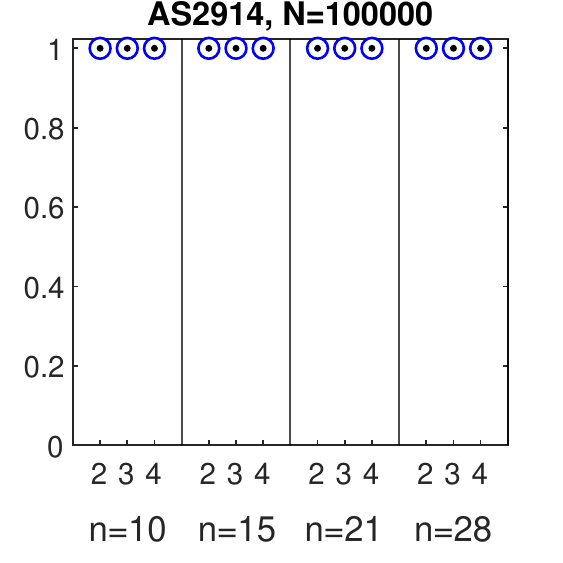}
		\caption{Recalls of the bounding topology estimates.}
		\label{fig:bounding-topo-rec-full}
	\end{figure}

	\subsection{Evaluating the Estimated Routing Matrix}
	
	Next, we evaluate the performance of Sparse M\"obius Inference end-to-end. We ran 
	stages 2 and 3 to get an estimate of $\hat{\mathbf R}$ for each case study and from 
	each of the three differently-sized samples, using as input to Stage 2 the bounding 
	topologies computed with $i_f = 4$ from the same sample. We tuned the hyperparameters 
	of the lasso heuristic ($\lambda$ and the exponent $b$) separately for each 
	underlying network, number of monitor paths, and value of $i_{\rm max}$. Tuning was 
	done by a grid search over values of $b$ from 0 to 1 in increments of 0.1, and values 
	of $\lambda$ from 0 to 4 in increments of 0.2. We then selected the pair of 
	parameters resulting in the highest average F1 score across the 10 cases studies with 
	the given underlying network and monitor paths, when evaluated with the given $i_{\rm 
	max}$ on $N = 100,000$. 
	
	
	To assess the accuracy of an estimated routing matrix $\hat{\mathbf R}$ compared to 
	the ground truth
	$\mathbf R$, we computed precision and recall in the following manner: precision is 
	the fraction of the columns of $\hat{\mathbf R}$ that are also 
	columns of $\mathbf R$, and recall is defined \textit{vice versa}. Then we combined 
	these two metrics into a single \textit{F1 score}, which is the geometric mean of 
	precision and recall. Figures \ref{fig:f1-100000}, \ref{fig:f1-50000}, and 
	\ref{fig:f1-10000} report the distributions of F1 scores that we obtained based on 
	samples of size 100,000, 50,000, and 10,000, respectively.
	
	We now examine how the performance of the inference depends on particular parameters:
	
	\paragraph*{Choice of $i_{\rm max}$}
	
	For most of our case studies with $N = 50,000$ or $N = 100,000$, the optimal choice 
	of $k$-statistic order is simply $i_{\rm max} = 3$. This is 
	likely because the third-order statistics contain more information than the 
	second-order statistics, but the fourth-order statistics were too noisy for the 
	solution of the lasso optimization problem to settle close enough to the true 
	cumulant values. (This is probably also the same reason why $i_{\rm max} = 2$ yields 
	the best performance when $N =10,000$.)
	
	The 8-monitor-node cases 
	from AS1221 are a notable exception; in most of these scenarios, $i_{\rm max} = 4$ 
	leads to the best performance. The likely 
	reason for this anomaly is visible in Figure \ref{fig:sparsity}: the AS1221 network, 
	with $n = 28$ monitor paths (corresponding to 8 monitor nodes), leads to the largest 
	size of $\supp(\mathbf f_i)$ and the largest path sets in $\supp(\mathbf f_i)$ among 
	all of the case studies. Thus, it is to be expected that common cumulant estimates 
	from more and larger path sets are necessary to identify which path sets in 
	$\supp(\mathbf f_i)$ also have a nonzero exact cumulant. These results suggest that 
	higher choices of $i_{\rm max}$ are needed when there is a greater amount of link 
	sharing between monitor paths, and in these cases, larger samples are needed to 
	ensure that accurate cumulant estimates are possible with limited uncertainty.
	
	\paragraph*{Choice of Sample Size}
	
	There were a few cases in which $N = 10,000$ were enough for decent performance 
	(AS1755 with 5 monitor nodes, and AS2914 with 5 or 6 monitor nodes). In all of these 
	cases, $i_{\rm max} = 2$ was also the optimal $k$-statistic order, indicating that 
	the small sample size is only sufficient when covariances alone contain nearly enough 
	information to reconstruct $\mathbf R$. Unfortunately, the sample size needs to be at 
	least 5 times larger in most other cases, so as to allow for an accurate estimate of 
	third-order cumulants. It is also interesting to observe that doubling the sample 
	size from 50,000 to 100,000 does not lead to significant improvements in the accuracy 
	of $\hat{\mathbf R}$.   
	
	\begin{figure}
		\centering
		{\large F1 Scores of $\mathbf{\hat R}$ from $N = 100,000$}
		\includegraphics[width=\linewidth]{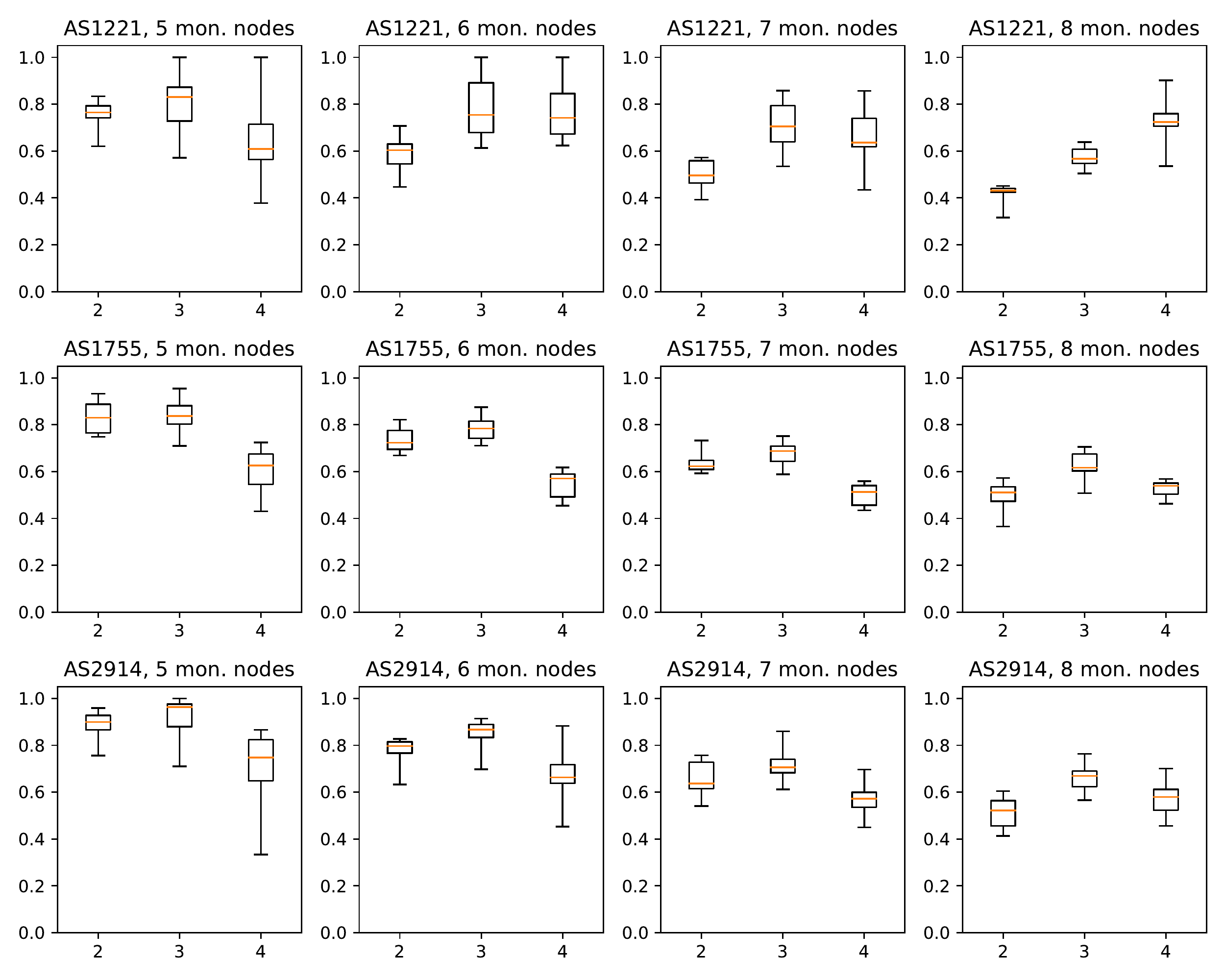}
		\caption{Distributions of F1 scores of the routing matrix estimate for the 120 
		case studies, based on a sample of size 100,000. Plots in each row are based on 
		the same underlying network, and plots in the same column have the same number of 
		monitor nodes. The three boxes in each plot correspond to values $i_{\rm max} = 
		2, 3, 4$ used for inference.}
		\label{fig:f1-100000}
	\end{figure}

	\begin{figure}
		\centering
		{\large F1 Scores of $\mathbf{\hat R}$ from $N = 50,000$}
		\includegraphics[width=\linewidth]{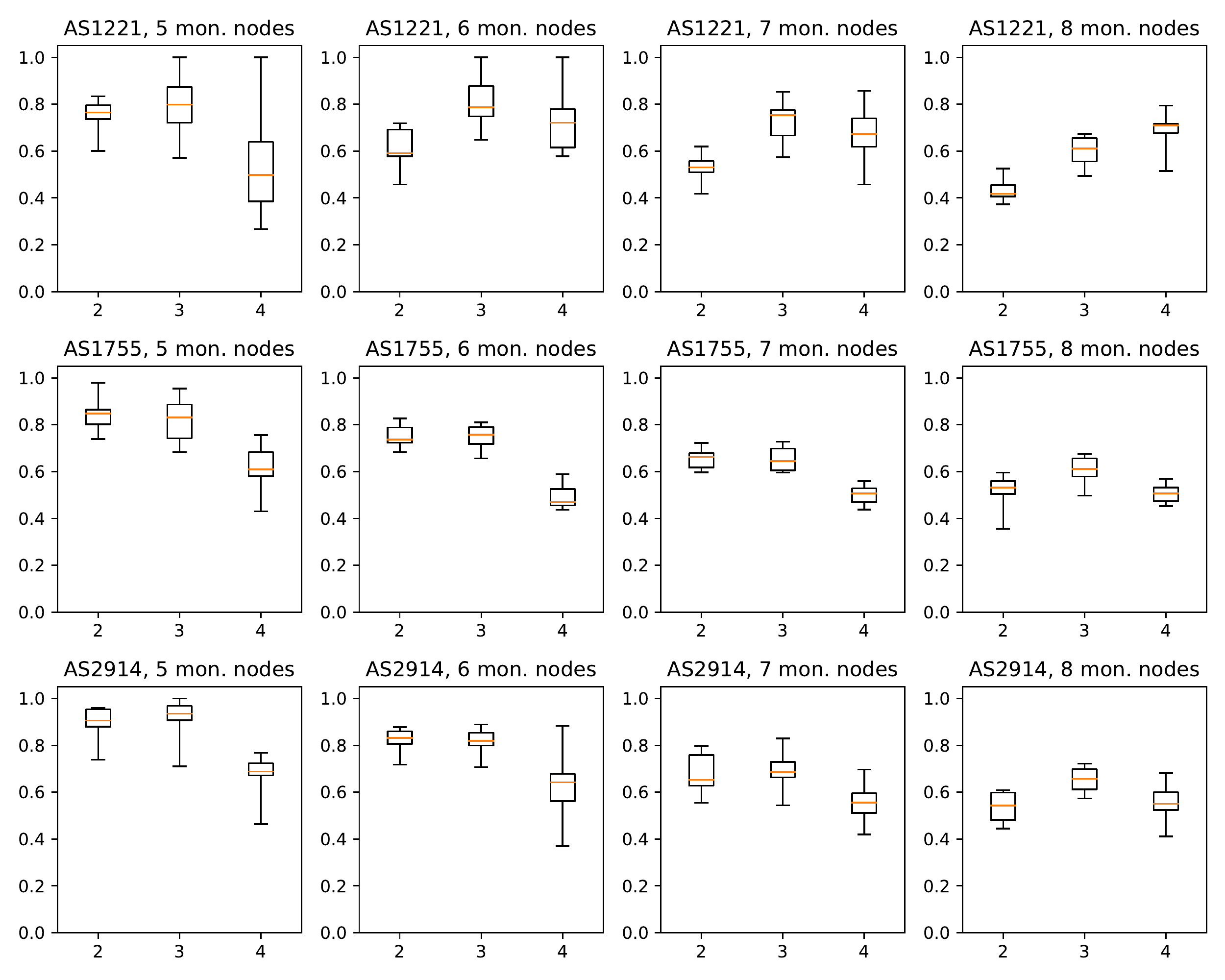}
		\caption{Distributions of F1 scores of the routing matrix estimate for the 120 
		case studies, based on a sample of size 50,000. Plots in each row are based on 
		the same underlying network, and plots in the same column have the same number of 
		monitor nodes. The three boxes in each plot correspond to values $i_{\rm max} = 
		2, 3, 4$ used for inference.}
		\label{fig:f1-50000}
	\end{figure}

	\begin{figure}
		\centering
		{\large F1 Scores of $\mathbf{\hat R}$ from $N = 10,000$}
		\includegraphics[width=\linewidth]{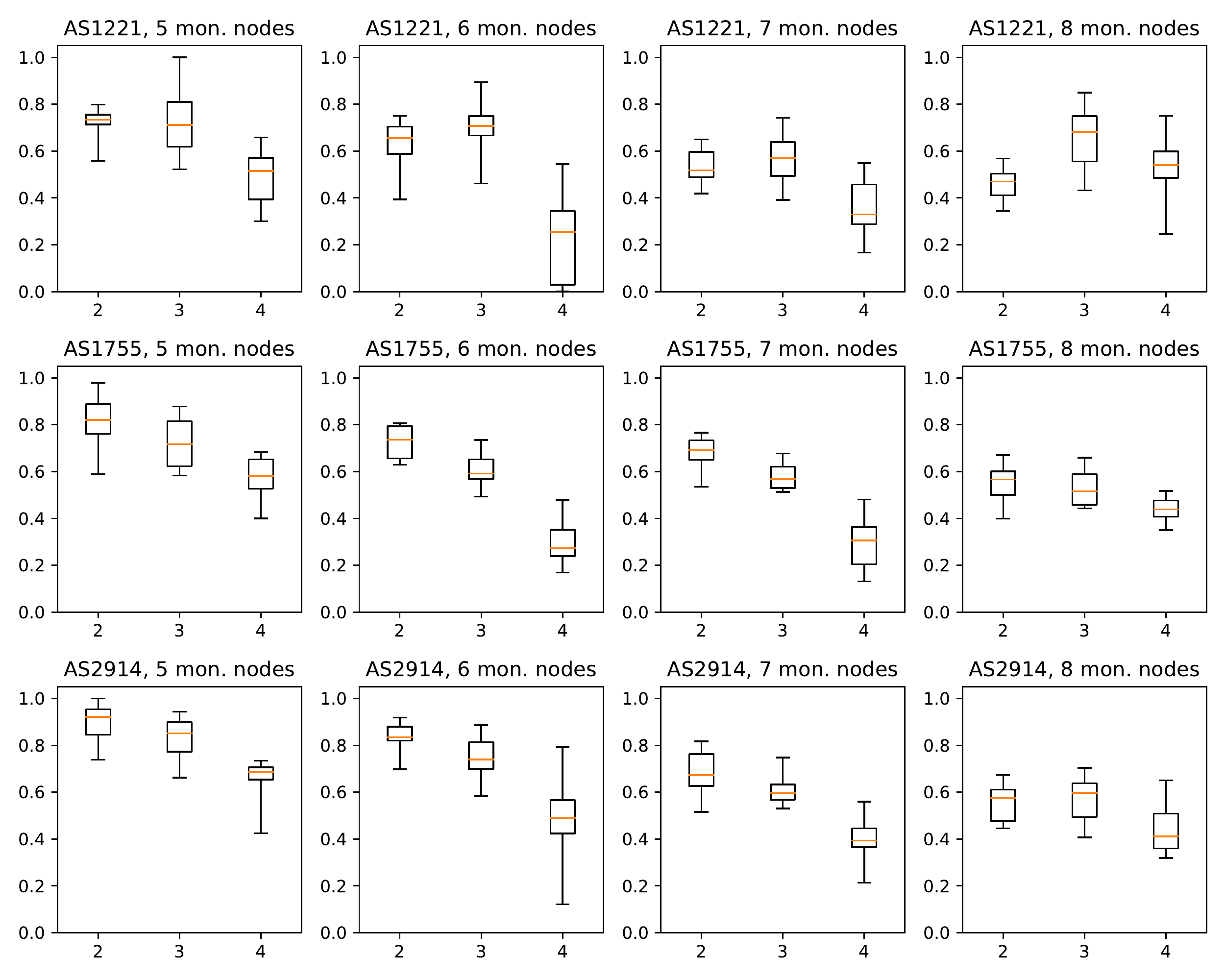}
		\caption{Distributions of F1 scores of the routing matrix estimate for the 120 
		case studies, based on a sample of size 10,000. Plots in each row are based on 
		the same underlying network, and plots in the same column have the same number of 
		monitor nodes. The three boxes in each plot correspond to values $i_{\rm max} = 
		2, 3, 4$ used for inference.}
		\label{fig:f1-10000}
	\end{figure}

	\section{Proofs from Section III}
	
	This section contains proofs from lemmas in Section III of the main manuscript. For 
	convenience, the lemma statements are reproduced as well.
	
	\subsection{Estimation Stage}
		
	\begin{lemma}[Properties of the Estimation Stage] 
		The following are true:
		\begin{enumerate}
			\item \label{item:rep-midx-count}
			Let $\rev{P \subseteq \monpaths}$. If $\rev{i} \ge |P|$, there are 
			$\binom{\rev{i} - 1}{|P| - 1}$ \rev{$i$th-order} representative multi-indices 
			of $P$.
			\item \label{item:f}
			For all $\rev{i \in \mathbb Z_{> 0}}$, the common cumulant $f_{\rev{i}}: 
			2^{\rev{\monpaths}} \to \mathbb R$ satisfies 
			\[
				f_n(P) = \sum_{\ell \in C(P)} \kappa_n(U_\ell), \qquad \forall P 
				\subseteq \monpaths
			\]
			\item \label{item:mia-dist-estimation-proof}
			Statement (ii) of Theorem 1 is true, i.e., Algorithm 1 correctly computes the 
			common cumulant vector for order $\rev{i} = n$.
		\end{enumerate} 
	\end{lemma} 
	
	\begin{proof}
		To prove \ref{item:rep-midx-count}, we will count the number of ways that 
		$\rev{i}$ ``counts'' of multiplicity can be assigned to the support of a 
		representative 
		multi-index. Each element of $P$ contains at least one count, and we are free to 
		\rev{distribute} the remaining $\rev{i} - |P|$ counts arbitrarily across the 
		elements 
		of $P$. Thus, there are $\multiset{|P|}{\rev{i} - |P|}$ ways to distribute the 
		remaining 
		counts, which is equivalent to $\binom{\rev{i} - 1}{|P| - 1}$.  
		
		To prove \ref{item:f}, let $\alpha$ be some $\rev{i}$th-order representative 
		multi-index of $P$. Using the independence of $U_\ell$ and the multilinearity of 
		multivariate cumulants, we have
		\begin{align*}
			f_{\rev{i}}(P) &= \kappa \left(
			\underbrace{ \mathbf R^{(1)} \mathbf U, \dots, \mathbf R^{(1)} \mathbf 
				U}_{\alpha(1) 
				~\text{times}}, 
			\dots, \underbrace{ \mathbf R^{(n)} \mathbf U, \dots, \mathbf R^{(n)} 
				\mathbf U}_{\alpha(n) 
				~\text{times}}
			\right) \\ 
			&= \sum_{\ell = 1}^m \left( r_{1 \ell}^{\alpha(1)} \cdots r_{n 
				\ell}^{\alpha(n)} \right) \kappa\left(
			\underbrace{U_\ell, \dots, U_\ell}_{\alpha(1) + \cdots + \alpha(n) 
				~\text{times}}
			\right) \\
			&= \sum_{\ell = 1}^m \left( \prod_{\rev{j} \in \supp(\alpha)} r_{\rev{j} 
				\ell} \right) 
			\kappa_{\rev{i}}(U_\ell)
		\end{align*} 
		where $\mathbf R^{(\rev{j})}$ denotes the $\rev{j}$th row of $\mathbf R$. Since 
		$\prod_{\rev{j} \in 
			\supp(\alpha)} r_{\rev{j} \ell} = 1$ if $\ell \in \rev{C}(P)$ and is zero 
		otherwise, we obtain
		\[
			f_i(P) = \sum_{\ell \in C(P)} \kappa_n(U_\ell), \qquad \forall P 
			\subseteq \monpaths
		\]
		
		To prove \ref{item:mia-dist-estimation-proof}, observe that the estimation stage 
		of Algorithm 1 defines the map $f_n$ precisely according to 
		Definition 3, so that $f_n$ is the common cumulant vector by line 
		6 of the algorithm. \rev{Then statement \ref{item:mia-dist-estimation-proof} 
		follows} by statement \ref{item:f} of this lemma. 
	\end{proof}

	\subsection{Inversion Stage}

	\begin{lemma}[Properties of the Inversion Stage] 
		Let $f_{\rev{i}}$ be the common cumulant vector, and let $g_{\rev{i}}: 
		2^{\rev{\monpaths}} \to \mathbb R$. The following three statements are equivalent:
		\begin{enumerate}
			\item \label{item:mia-dist-inversion-exact}
			$g_{\rev{i}}$ is the exact cumulant vector.
			\item \label{item:mia-dist-inversion-pre}
			$f_{\rev{i}}$ and $g_{\rev{i}}$ satisfy
			\begin{equation} \label{eq:pre-inversion}
				f_{\rev{i}}(P) = \sum_{Q \supseteq P} g_{\rev{i}}(Q), \qquad \forall 
				\rev{P \subseteq \monpaths} \tag{B.1}
			\end{equation}
			\item \label{item:mia-dist-inversion-post}
			$f_{\rev{i}}$ and $g_{\rev{i}}$ satisfy
			\begin{equation}
				g_{\rev{i}}(P) = \sum_{Q \supseteq P} (-1)^{|Q| - |P|} 
				f_{\rev{i}}(Q), \qquad \forall \rev{P \subseteq \monpaths} \tag{B.2}
				\label{eq:inversion}
			\end{equation}
		\end{enumerate}
		Furthermore, statement (iii) of Theorem 1 is true, i.e., the Algorithm 1 
		correctly computes the exact cumulant vector.
	\end{lemma}
	
	\begin{proof}
		We begin with the equivalence \ref{item:mia-dist-inversion-pre} $\iff$ 
		\ref{item:mia-dist-inversion-post}. This equivalence holds for \textit{any} 
		functions $f_{\rev{i}}, g_{\rev{i}}: 2^{\rev{\monpaths}} \to \mathbb R$, and it 
		follows from the 
		M\"obius inversion formula \rev{applied over $2^{\monpaths}$.} See, for 
		example, \cite[Theorem 5.1]{MA:07b}.
		
		To prove that \ref{item:mia-dist-inversion-exact} $\implies$
		\ref{item:mia-dist-inversion-pre}, we will first show that
		\begin{equation}
			\rev{C}(P) = \bigcup_{Q \supseteq P} \rev{E}(\rev{Q}) \tag{B.3}
			\label{eq:LM}
		\end{equation}
		Let $\ell \in \rev{C}(P)$, and examine the column of the routing matrix $\mathbf 
		R_\ell \in \{0, 1\}^n$. There is some $\rev{Q \subseteq \monpaths}$ for which the 
		characteristic vector satisfies $\rev{\chi(Q, \monpaths)} = \mathbf R_\ell$. It 
		follows that $\ell \in \rev{E}(Q)$. Now, because $\ell \in \rev{C}(P)$, it 
		follows that $r_{p \ell} = 1$ for all $p \in P$, so that $Q \supseteq P$. 
		Therefore $\ell \in \bigcup_{Q \supseteq P} 
		\rev{E}(Q)$. Next, let $\ell \in \bigcup_{Q \supseteq P} \rev{E}(Q)$, so that 
		$\ell \in \rev{E}(Q)$ for some $Q \supseteq P$. It is clear that $r_{p \ell} = 1$ 
		for all $p \in Q$, so the inclusion $Q \supseteq P$ implies that $\ell \in 
		\rev{C}(P)$. Now, if $g_{\rev{i}}$ is the exact cumulant 
		vector, we can (\rev{from Definition 5}) substitute \eqref{eq:LM} into 
		\begin{equation}
			g_n(P) = \sum_{\ell \in E(P)} \kappa_n(U_\ell), \qquad \forall P \subseteq 
			\monpaths
			\label{eq:g} \tag{B.4}
		\end{equation}
		obtaining
		\[
		\sum_{Q \supseteq P} g_{\rev{i}}(Q)
		= \sum_{Q \supseteq P} \sum_{\ell \in \rev{E}(Q)} 
		\kappa_{\rev{i}}(U_\ell) 
		= \sum_{\ell \in \rev{C}(P)} \kappa_{\rev{i}}(U_\ell) 
		= f_{\rev{i}}(P)
		\]
		The last step follows from Lemma \ref{lem:mia-dist-estimation} \ref{item:f}. 
		Hence \ref{item:mia-dist-inversion-exact} $\implies$ 
		\ref{item:mia-dist-inversion-pre}. 
		
		To prove that \ref{item:mia-dist-inversion-pre} $\implies$ 
		\ref{item:mia-dist-inversion-exact}, suppose that $f_{\rev{i}}$ and $g_{\rev{i}}$ 
		satisfy \eqref{eq:pre-inversion}. By \eqref{eq:LM},
		\begin{equation}
			\sum_{Q \supseteq P} g_{\rev{i}}(Q) 
			= \sum_{Q \supseteq P} \sum_{\ell \in \rev{E}(Q)} 
			\kappa_{\rev{i}}(U_\ell)
			\label{eq:recover-g} \tag{B.5}
		\end{equation}
		for all $\rev{P \subseteq \monpaths}$. We will use \eqref{eq:recover-g} to show 
		that $g_{\rev{i}}$ satisfies \eqref{eq:g} by strong induction over $|P|$. In the 
		$|P| = n$ base case, the only possible set is $P = \rev{\monpaths}$, for which 
		\eqref{eq:recover-g} reduces to $g_{\rev{i}}(\rev{\monpaths}) = \sum_{\ell \in 
			\rev{E}(\rev{\monpaths})} \kappa_{\rev{i}}(U_\ell)$. Now suppose that 
		\eqref{eq:g} holds for all $P$ with $|P| 
		\rev{\ge i}$ for some $\rev{j} \in [2, n]$. Let $\rev{P \subseteq \monpaths}$ 
		such that $|P| = \rev{j} - 1$, and observe that
		\[
		\sum_{Q \supseteq P} g_{\rev{i}}(Q) 
		= g_{\rev{i}}(P) + \sum_{Q \supset P} \sum_{\ell \in \rev{E}(Q)} 
		\kappa_{\rev{i}}(U_\ell)
		\]
		by the inductive hypothesis. Substituting this equation in to 
		\eqref{eq:recover-g} and simplifying, we obtain \eqref{eq:g}. Hence \eqref{eq:g} 
		holds for all $\rev{P \subseteq \monpaths}$, so \ref{item:mia-dist-inversion-pre} 
		$\implies$ \ref{item:mia-dist-inversion-exact}.
		
		To prove the final statement, note that the inversion stage of Algorithm 
		1 defines the map $g_n$ according to \eqref{eq:inversion}, where 
		$f_n$ is the common cumulant vector (per Lemma 4 
		\ref{item:mia-dist-estimation-proof}), by line 10. It follows from 
		the equivalence proven in this lemma that $g_n$ is the exact cumulant vector.
	\end{proof}

	\subsection{Reconstruction Stage}
	
		\begin{lemma}[Properties of the Reconstruction Stage] 
		Let $g_n: 2^{\rev{\monpaths}} \to \mathbb R$ be the exact cumulant vector. For 
		each $\rev{P \subseteq \monpaths}$, let $\rev{\chi(P, \monpaths)} \in \{0, 
		1\}^n$ be the characteristic vector of $P$ in $\rev{\monpaths}$. The following 
		are true:
		\begin{enumerate}
			\item
			If $P \in \supp(g_n)$, then 
			$\rev{\chi(P, \monpaths)}$ must be a column of the routing matrix. Under 
			Assumptions 1 and 2, the converse is also true. 
			\item
			Statement (iv) of Theorem 1 is true. 
		\end{enumerate}
	\end{lemma} 
	
	\begin{proof}
		If $g_n(P) \ne 0$, it is clear from \eqref{eq:g} that $\rev{E}(P)$ is non-empty, 
		which implies that some column of the routing matrix $\mathbf R_\ell$ satisfies 
		$\rev{\chi(P, \monpaths)} = \mathbf R_\ell$. Now suppose that Assumptions 
		1 and 2 are true. By Assumption 1, the set $\rev{E}(P)$ is 
		either empty or contains a single element. By Assumption 2, if 
		$\rev{E}(P)$ contains a single element $\ell$, it must satisfy $\kappa_n(U_\ell) 
		\ne 0$. Therefore, if $g_n(P) = 0$, under these two assumptions, it follows that 
		$\rev{E}(P)$ is empty. Hence $\rev{\chi(P, \monpaths)}$ is not a column of the 
		routing matrix.
		
		Per Lemma 6, the vector $g_n$ in Algorithm 
		1 is the exact cumulant vector by line 10, so we can apply the above result to 
		$g_n$ in the reconstruction stage of the algorithm, yielding statement (iv) of 
		Theorem 1.
	\end{proof}

	\section{Proofs from Section V}
	
	This section contains proofs from statements in Section V of the main manuscript. For 
	convenience, the statements are reproduced as well.
	
	\begin{theorem}[Properties of Algorithm 2]
		Let $\mathcal B \subseteq 2^{\monpaths}$ be a collection of path sets, let $i \in 
		\mathbb{Z}_{> 0}$ be a cumulant order, and let $t: \mathbb{Z}_{> 0} \times 
		\mathbb{Z}_{> 0} \to \mathbb{Z}_{> 0}$ be a 
		threshold function. The following are true:
		\begin{enumerate}
			\item Algorithm 2 evaluates $\texttt{IsNonzero}(f_i(P))$ 
			$O(n^i)$ times and terminates after $O(2^q)$ iterations of the while loop, 
			where $q$ is the size of the largest set in $\mathcal B$. The algorithm 
			returns a collection of path sets $\mathcal B' \subseteq 2^{\monpaths}$. 
			\item The support estimate of $\mathcal B'$ is a subset of the support 
			estimate of $\mathcal B$. 
			\item For any set $P$ in the support estimate of $\mathcal B$, $P$ is also in 
			the support estimate of $\mathcal B$ if either $|P| < i$, or if there is a 
			superset $Q \supseteq P$ in the support estimate of $\mathcal B$ for which at 
			least $t(|Q|, i)$ size-$i$ subsets $R \subseteq Q$ satisfy 
			$\texttt{Nonzero}(f_i(R))$. 
		\end{enumerate}
	\end{theorem}
	
	\begin{proof}
		There are at most $\binom{n}{i} = O(n^i)$ size-$i$ sets, so 
		$\texttt{Nonzero}(f_i(P))$ is evaluated $O(n^i)$ times to compute $\mathcal P$. 
		The worst-case runtime occurs when $|\{P \in \mathcal P : P \subseteq B\}| < 
		t(|B|, i)$ for each iteration of the while loop, in which case the variable $B$ 
		takes on the value of every subset (with size at least $i$) of every original set 
		in $\mathcal B$ precisely once (because the collection $\mathcal X$ tracks which 
		sets have already been processed, preventing redundant iterations of the while 
		loop). Thus, there are $O(2^q)$ iterations of the while loop. 
		
		To prove (ii), observe that every set added to $\mathcal B'$ was originally in 
		the queue $\mathcal B$, and that sets in the queue are either from the original 
		collection $\mathcal B$, or they are subsets of a previous element in the queue. 
		Hence every set in $\mathcal B'$ is a subset of a set in the original $\mathcal 
		B$, so the support estimate of $\mathcal B'$ is a subset of the original support 
		estimate. To prove (iii), suppose that $P$ is in the support estimate of 
		$\mathcal B'$, so that some $B' \in \mathcal B'$ contains $P$. Sets are only 
		added to $\mathcal B'$ on line 5, and the set must satisfy either 
		$|B'| < i$ or $|\{P' \in \mathcal P : P' \subseteq B'\}| \ge t(|B'|, i)$, i.e., 
		(b) is satisfied with $Q = B'$.  
	\end{proof}

\begin{lemma}[Elimination of Large, Non-Maximal Path Sets]
	Let $\mathcal B \subseteq 2^{\monpaths}$ be a collection of path sets, and let $s 
	\in \mathbb{Z}_{> 0}$. Assume that the following are true:
	\begin{enumerate}
		\item Every set in $\mathcal B$ is maximal (i.e., no $B, B' \in \mathcal B$ 
		exist such that $B \subset B'$), \label{cond:maximal}
		\item $f_i(P) \ne 0$ and $g_i(P) \ne 0$ only if $P$ is in the support 
		estimate of $\mathcal B$, and
		\item $g_i(P) = 0$ for all $P \subseteq \monpaths$ with $|P| > s$ and $P 
		\notin \mathcal B$. 
		\label{cond:heuristic}
	\end{enumerate}
	Then for every $P$ in the support estimate of $\mathcal B$ such that $|P| \le s$, 
	\begin{align*}
		\begin{split}
			g_i(P) &= \sum_{Q \supseteq P : |Q| \le s} (-1)^{|Q| - |P|} f_i(Q) \\
			&\qquad - \sum_{B \in \mathcal B : B \supseteq P} (-1)^{s - |P|} \binom{|B| - 
			|P| 
				- 1}{s - |P|} f_i(B)  
		\end{split} 
	\end{align*}
\end{lemma}

\begin{proof}
	Let $P$ be in the support estimate of $\mathcal B$ with $|P| \le s$. We can split 
	the M\"obius inversion formula into two parts:
	\begin{align*}
		g_i(P) &= \sum_{Q \supseteq P : |Q| \le s} (-1)^{|Q| - |P|} f_i(Q) \\
		&\qquad + \sum_{R \supseteq P : |R| > s} (-1)^{|R| - |P|} f_i(R)
	\end{align*}
	Focus on the second sum, and let $R \supseteq P$ such that $|R| > s$. Condition 
	\ref{cond:heuristic} implies that
	\[
	f_i(R) = \sum_{Q \supseteq R} g_i(Q) = \sum_{B \in \mathcal B : B \supseteq 
		R} g_i(B)
	\]
	so we can simplify the second sum by
	\begin{align*}
		&\sum_{R \supseteq P : |R| > s} (-1)^{|R| - |P|} f_i(R) \\
		&\qquad = \sum_{R \supseteq P : |R| > s} (-1)^{|R| - |P|} \sum_{B \in \mathcal B: 
			B 
			\supseteq R} g_i(B) \\
		&\qquad = (-1)^{-|P|} \sum_{B \in \mathcal B} g_i(B) \sum_{B 
			\supseteq R \supseteq P : |R| > s} (-1)^{|R|} \\
		&\qquad = (-1)^{-|P|} \sum_{B \in \mathcal B} g_i(B) 
		\sum_{j=s+1}^{|B|} (-1)^j \binom{|B| - |P|}{j - |P|} \\
		&\qquad = \sum_{B \in \mathcal B} (-1)^{s + 1 - |P|} \binom{|B| - 
			|P| - 1}{s - |P|} g_i(B)
	\end{align*}
	Finally, observe that $g_i(B) = f_i(B)$, since there are no proper supsersets of 
	$B$ in the support estimate of $\mathcal B$ (due to condition 
	\ref{cond:maximal}).
\end{proof}

	\bibliographystyle{IEEEtran}
	\bibliography{alias,New,Main,FB}